\documentclass{LMCS}
\def\chg{\empty}
\def\endchg{\empty}
\newcommand{\myvc}[1]{\vcenter{\hbox{\ensuremath{#1}}}}

\newcommand{\comment}[1]{\empty}

\usepackage[all]{xy}
\SelectTips{cm}{}
\def\subto{\hookrightarrow}
\def\refeq#1{(\ref{#1})}
\def\ol{\overline}

\usepackage{amsmath}
\usepackage{amssymb}
\usepackage{mathrsfs,euscript}
\usepackage{pstricks,pst-node,pst-tree,pst-coil}
\usepackage{enumerate,hyperref}

\newcommand{\Set}{\ensuremath{\St}}
\newcommand{\FF}{\ensuremath{\mathscr{F}}}
\newcommand{\KK}{\ensuremath{\mathscr{K}}}
\newcommand{\WW}{\ensuremath{\mathscr{W}}}
\newcommand{\AAA}{\ensuremath{\mathscr{A}}}
\newcommand{\SetF}{\ensuremath{\St^{\mathscr{F}}}}

\newcommand{\DD}{\ensuremath{\mathscr{D}}}
\newcommand{\fr}{\text{fr}}
\newcommand{\Gfr}{\ensuremath{\Gamma_{\kern-.5pt\fr}}}
\newcommand{\rr}{\ensuremath{\mathbb{R}}}
\newcommand{\nn}{\ensuremath{\mathbb{N}}}
\newcommand{\MM}{\ensuremath{\mathbb{M}}}
\newcommand{\TT}{\ensuremath{\mathbb{T}}}
\newcommand{\Sig}{\ensuremath{\Sigma}}
\newcommand{\tecd}{{.}}
\newcommand{\tec}{{\cdot}}
\newcommand{\nsi}[1]{\xrightarrow{\rule{1mm}{0mm}#1\rule{1mm}{0mm}}}
\newcommand{\lnsi}[1]{\xleftarrow{\rule{1mm}{0mm}#1\rule{1mm}{0mm}}}
\newcommand{\zav}{\mathbin{@}}
\newcommand{\ctv}{\mathbin{\raisebox{.5pt}{%
$\scriptstyle\Box$}}}

\DeclareMathOperator{\op}{op}
\DeclareMathOperator{\card}{card}
\DeclareMathOperator{\inr}{\mathsf{inr}}
\DeclareMathOperator{\inl}{\mathsf{inl}}
\DeclareMathOperator{\injection}{\mathsf{in}}
\DeclareMathOperator{\fold}{\mathsf{fold}}
\DeclareMathOperator{\unfold}{\mathsf{unfold}}
\DeclareMathOperator{\curry}{\mathsf{curry}}
\DeclareMathOperator{\obj}{obj}
\DeclareMathOperator{\can}{\mathsf{can}}
\DeclareMathOperator{\Fin}{\textnormal{Fin}}
\DeclareMathOperator{\St}{\textnormal{\textbf{Set}}}
\DeclareMathOperator{\id}{id}
\DeclareMathOperator{\Id}{Id}
\DeclareMathOperator{\Alg}{\textnormal{\textbf{Alg}}}

\DeclareMathOperator{\cpo}{\textnormal{CPO}}
\newcommand{\CPO}{\ensuremath{\cpo}}
\DeclareMathOperator{\cpob}{\textnormal{\textbf{CPO}}}
\DeclareMathOperator*{\colim}{colim}
\DeclareMathOperator{\Kan}{Kan}
\DeclareMathOperator{\Var}{\mathsf{Var}}

\def\spitze#1{\langle #1,#1\rangle}
\def\sem#1{[\! [ #1]\!]}

\theoremstyle{definition}
\newtheorem{notation}[thm]{Notation}

{\qed\end{trivlist}}

{\qed\end{trivlist}}

\makeatletter
\edef\pst@arrowtable{\pst@arrowtable,H-H}
\def\tx@RHook{RHook }
\@namedef{psas@H}{%
   /RHook {
     CLW mul add dup CLW sub 2 div
     /x ED mul
     /y ED
     /z CLW 2 div def
     x neg y moveto
     x neg 0 0 0 0 y
     curveto
     stroke 0 y moveto %
   } def
   \psk@bracketlength\space \psk@tbarsize\space \tx@RHook%
}
\makeatother

\def\doi{7 (1:15) 2011}
\lmcsheading%
{\doi}
{1--43}
{}
{}
{Jan.~11, 2010}
{Mar.~31, 2011}
{}

\begin{document}

\title{Semantics of Higher-Order Recursion Schemes}

\author[J.~Ad\'amek]{Ji\v r\'\i\ Ad\'amek\rsuper a}
\address{{\lsuper{a,b}}Institut f\"ur Theoretische Informatik, Technische
Universit\"at Braunschweig, Germany}
\email{adamek@iti.cs.tu-bs.de, mail@stefan-milius.eu}

\author[S.~Milius]{Stefan Milius\rsuper b}
\address{\vskip-6 pt}

\author[J.~Velebil]{Ji\v r\'\i\ Velebil\rsuper c}
\address{{\lsuper c}Faculty of Electrical Engineering, Czech
Technical University of Prague, Prague, Czech Republic}
\email{velebil@math.feld.cvut.cz}
\thanks{{\lsuper c}Supported
by the grant MSM 6840770014 of the Ministry of Education of
the Czech Republic.}


\begin{abstract}
\noindent
Higher-order recursion schemes are recursive equations defining
new operations from given ones called ``terminals''. Every such recursion
scheme is proved to have a least interpreted semantics in
every Scott's model of $\lambda$-cal\-cu\-lus in which the terminals
are interpreted as continuous operations. For the uninterpreted
semantics based on infinite $\lambda$-terms we follow the idea of
Fiore, Plotkin and Turi and work in the category of sets in context,
which are presheaves on the category of finite sets.
\chg
Fiore \textit{et al} showed how to capture the type of variable
binding in $\lambda$-calculus by an endofunctor $H_\lambda$
and they explained simultaneous substitution of
$\lambda$-terms by proving that the
presheaf of $\lambda$-terms is an initial
$H_\lambda$-monoid. Here we work with the presheaf of
\endchg
rational infinite $\lambda$-terms and prove that this is an initial
iterative $H_\lambda$-monoid.
We conclude that every guarded higher-order recursion
scheme has a unique uninterpreted solution in this monoid.
\end{abstract}

\maketitle

\section{Introduction}

\noindent The present paper is a contribution to the study of the
semantics of recursive definitions using category-theoretic tools and
methods. Our goal is to present a category-theoretic semantics of
higher-order recursion schemes in the sense of W.~Damm~\cite{D}. To
reach this goal we apply the theory of rational monads on a
category~\KK, developed in our previous work~\cite{AMV1} in order to
formalize iteration in algebra, to the category
\[\KK=\SetF\qquad(\FF=\text{finite sets and functions})\]
of sets in context. We use the approach to $\lambda$-calculus based
on $H$-monoids in the category of sets in context due to M.~Fiore,
G.~Plotkin and D.~Turi~\cite{FPT}. Our main result is a description
of the initial
iterative $H$-monoid as the monoid of rational
$\lambda$-terms, and the fact that in this monoid every higher-order
recursion scheme has a unique uninterpreted solution.

We now explain the motivation of our paper in more detail. In the
higher-order semantics we assume a given collection~\Sig\
of existing programs of given types (that is, a
many-sorted signature of ``terminals''). One recursively
defines new typed
programs $p_1,\dots, p_n$ (forming a many-sorted signature of
``nonterminals'') using symbols from $\Sig$ and $\{p_1,\dots,p_n\}$.
If the
recursion only concerns application, we can formalize this
as a collection of equations
\begin{equation}
\label{hrjj}
p_i=f_i\qquad(i=1,\dots,n)
\end{equation}
whose right-hand sides~$f_i$ are terms in the signature of all
terminals and all non-terminals. Such collections are called
(first-order) \textit{recursion schemes} and were studied in 1970's by
various authors, e.g. B.~Courcelle, M.~Nivat and I.~Guessarian (see
the monograph~\cite{G} and references there) or S.~J.~Garland and
D.~C.~Luckham~\cite{GL}. Recently, a categorical approach to semantics
of first-order recursion schemes was presented by S.~Milius and
L.~Moss~\cite{MM}. In the present paper we take a first step in an
analogous approach to the semantics of \textit{higher-order recursion
  schemes} in which $\lambda$-abstraction is also used as one of the
operations. That is, a higher-order recursion scheme, as introduced by
W.~Damm~\cite{D} (see also the recent contributions \cite{Ae}
and~\cite{M}) is a collection of equations $p_i=f_i$ where $f_i$~are
terms using application and $\lambda$-abstraction on symbols
from~$\Sig$
and $\{p_1,\dots,p_n\}$. As in~\cite{MM}, we first study the
uninterpreted semantics, where the given system is regarded as a
purely syntactic construct. At this stage the operation symbols in
$\Sigma$ as well as $\lambda$-abstraction and application have no
interpretation on actual data. So the semantics is provided by formal
(infinite) terms. These terms can be represented by rational trees,
i.\,e., infinite trees having finitely many subtrees. Thus the
uninterpreted solution assigns to each of the recursive variables
$p_i$ in~(\ref{hrjj}) a rational tree $p_i^\dag$ such that the formal
equations become identities if we substitute~$p_i^\dag$ for~$p_i$
($i=1,\dots,n$). We assume $\alpha$-conversion (renaming of bound
variables) but no other rules in the uninterpreted semantics. We next
turn to an interpreted semantics. Here a recursion scheme is given
together with an interpretation of all symbols from $\Sigma$ as well
as $\lambda$-abstraction and application.
\chg
Following D.~Scott,
we interpret the $\lambda$-calculus on a CPO, say
$D$.
\endchg
The
symbols of $\Sigma$ are interpreted as continuous operations on~$D$,
and formal $\lambda$-abstraction and application are the actual
$\lambda$-abstraction and application in the model $D$. An interpreted
solution in~$D$ then assigns to each~$p_i$ in the context~$\Gamma$ of
all free variables in~\eqref{hrjj} an element of $\cpob(D^\Gamma,D)$
(continuously giving to each assignment of free variables in
$D^\Gamma$ an element of $D$) such that the formal equations in the
recursion scheme become identities in~$D$ when the right-hand sides
are interpreted in~$D$.

\begin{exa}
\label{jj}
The fixed-point operator~$Y$ is specified by
\[Y=\lambda f{.}f(Yf)\]
and the uninterpreted semantics is the rational tree
\begin{equation}
\label{rnjj}
Y^\dag=\quad\raisebox{25mm}{$
\pstree[nodesep=3pt,levelsep=10mm,treesep=18mm]{\TR{\lambda f}}{%
\pstree{\TR{\zav}}{%
\TR{f}\pstree{\TR{\zav}}{%
\pstree{\TR{\lambda f}}{%
\pstree{\TR{\zav}}{\TR{f}\TR{{\ddots}}}}\TR{f}}}}$}
\end{equation}
(The symbol~$\zav$ makes application explicit.)
The interpreted solution in $D$ is the least fixed point operator
(considered as an element of $D$).
\end{exa}

The above example is untyped, and indeed we are only treating the
untyped
case in the present paper since its uninterpreted semantics is
technically simpler than the
typed case; however, the basic ideas of uninterpreted semantics
are similar.
In contrast, the interpreted semantics (based on a
specified model of $\lambda$-calculus with ``terminal'' symbols
interpreted as operations) is more subtle in the untyped case.

Our main result is that every guarded higher-order recursion scheme
has a unique uninterpreted solution,
and a least interpreted one. This demonstrates that the methods for
iteration in locally finitely presentable categories developed
in~\cite{AMV1} can serve not only for first-order iteration,
when applied to endofunctors of $\Set$, but also for higher-order
iteration: it is possible to apply these methods to other
categories, here the category of sets in context.
\smallskip

\noindent
{\bf Related Work.}
This is an extended and revised version of the conference
paper~\cite{AMV2}. In addition to the material in that extended
abstract we include here the theory of
iterative monoids in a monoidal category, see Section~\ref{snc} below,
and we provide detailed proofs.

\section{Presheaves as Algebras}
\label{sd}
\setcounter{equation}{0}
\begin{notation}
\label{dj}\hfill
\begin{enumerate}[(1)]
\item
Throughout the paper a given countably infinite set~$\Var$ of
variables is assumed. Finite subsets $\Gamma\subseteq\Var$ are called
\textit{contexts} and form a full subcategory~\FF\ of~\Set. We also
assume that a (possibly empty) finitary signature~\Sig\ is given.

When speaking about formulas in context~$\Gamma$ we mean those that
have all free variables in~$\Gamma$. For example, $\lambda x\tecd yx$
is a formula in context $\Gamma=\{y,y'\}$.

\item
The category~\SetF\ of ``covariant presheaves'' on~\FF\
is well known to be equivalent to the
category
of finitary endofunctors of~\Set. Indeed, every endofunctor~$X$
yields the
presheaf~$X\restriction\FF$,
and conversely,
every presheaf~$X$ in~\SetF\ has
a left $\Kan$~extension to a finitary endofunctor of~\Set: for every
set~$M$ we have
\[X(M)=\bigcup Xi_\Gamma\bigl[X(\Gamma)\bigr]\]
where the union ranges over embeddings
$i_\Gamma\colon\Gamma\hookrightarrow M$ of contexts $\Gamma$ into~$M$,
and $Xi_\Gamma [X(\Gamma)]$ denotes the image of $Xi_\Gamma$.

\item
From now on we speak about presheaves when objects of~\SetF\ are
meant. The word
\textit{endofunctor} is
reserved for endofunctors on~\SetF\ throughout our
paper.
\end{enumerate}
\end{notation}

\begin{exa}
\label{dd}\hfill
\begin{enumerate}[(i)]
\item
The \textit{presheaf of variables}, $V$, is our name for the
embedding $\FF\hookrightarrow\Set$: $V(\Gamma)=\Gamma$. As we will
see in Section~\ref{ht}, $V$~is the unit of the monoidal operation of
substitution.

\item
\textit{Free presheaf} on one generator of context~$\Gamma$ is
our name for the representable presheaf
\[\FF(\Gamma,{-}).\]
Indeed, the Yoneda lemma states that this presheaf is freely
generated by the
element~$\id_\Gamma$ of context~$\Gamma$: for every presheaf~$X$ and
every $x\in X(\Gamma)$ there exists a unique morphism
$f\colon\FF(\Gamma, {-})\to
X$ with $f_\Gamma(\id_\Gamma)=x$. Observe that
$\FF(\Gamma,{-})$~is
naturally isomorphic to the functor $X\mapsto X^n$, where
$n=\card\Gamma$ is the power of~$\Gamma$. Consequently a free presheaf
on $k$~generators in contexts $\Gamma_1,\dots,\Gamma_k$ has the form
\[\Gamma\mapsto \Gamma^{n_1}+\dots+\Gamma^{n_k},\qquad
\text{where
$n_i=\card\Gamma_i$.}\]
This is the ``polynomial presheaf''~$X_{\Sig}$ of a signature~\Sig\
of
$k$~operation
symbols of the given arities~$n_i$.

\item

The \textit{presheaf~$F_\lambda$ of (finite)
$\lambda$-terms} is defined
via a quotient since we want to treat $\lambda$-terms always
modulo $\alpha$-conversion.
\chg
We first consider the set of all
$\lambda$-trees~$\tau$ given by the grammar
\begin{equation}
\label{add}
\tau\mathbin{{:}{:}{=}}x \mid\tau\zav\tau\mid\lambda
y\tecd\tau\qquad (x, y\in\Var).
\end{equation}
In the graphic form:\smallskip

\begin{equation}
\label{rdj}
\myvc{
\begin{pspicture}(0,0)
\pscirclebox[linewidth=.5pt,framesep=1.5pt]{$x$}
\end{pspicture}
}\qquad\text{or}\quad
\myvc{
\pstree[levelsep=10mm,treesep=10mm]{%
\TR{\raisebox{-3pt}{\rule{0pt}{12pt}}\zav}}{%
\pstree{\Tp~[tnpos=b,tnsep=5mm]{\tau}}{\Tfan}
\pstree{\Tp~[tnpos=b,tnsep=5mm]{\tau'}}{\Tfan}}
}\qquad\text{or}\qquad
\myvc{
\pstree[levelsep=10mm,treesep=10mm]{%
\TR{\raisebox{-3pt}{\rule{0pt}{12pt}}\lambda y}}{%
\pstree{\Tp~[tnpos=b,tnsep=5mm]{\tau}}{\Tfan}}
}
\end{equation}\smallskip

\noindent The notions of a free and bound variable of a $\lambda$-tree $\tau$
are defined as usual.

As explained in~\cite{FPT}, the following
approach is equivalent to defining $\lambda$-terms up to
$\alpha$-equivalence by de Bruijn levels:
We first denote by~$F'_\lambda(\Gamma)$ the set of all finite
$\lambda$-trees
with free variables in the context $\Gamma=\{x_1,\dots,x_n\}$.
\endchg
We then define the presheaf $F_\lambda$ in context $\Gamma$ by
\[F_\lambda(\Gamma)=F'_\lambda(\Gamma)/{\sim_\alpha}\]
where $\sim_\alpha$~represents
the $\alpha$-conversion: this is the least congruence with
$\lambda
y\tecd\tau\sim_\alpha\lambda
z\tecd\tau\bigl[\raisebox{2pt}{$z$}\big/
\raisebox{-2pt}{$y$}\bigr]$,
\chg where $z$ is not a free variable of $\tau$. And we
define $F_\lambda$ on morphisms $\gamma:\Gamma \to \Gamma'$ by
choosing a
term $t \in F_\lambda(\Gamma)$, relabelling all bound variables so
that they do not lie in $\Gamma'$, and denoting by
$F_\lambda\gamma(t)$ the
term obtained by relabelling every free variable $x \in \Gamma$ to
$\gamma(x) \in \Gamma'$.
\endchg

We call the congruence classes of finite $\lambda$-trees modulo
$\alpha$-conversion
\emph{finite
  $\lambda$-terms}.
\chg
(Finite $\lambda$-trees do not form a presheaf, due to possible
clashes
of
bound and free variables. For example consider the $\lambda$-tree
$$
\xy
\POS   (000,000) *+{\lambda x} = "r"
,      (000,-10) *+{\zav} = "a"
,      (-05,-20) *+{x}  = "x"
,      (005,-20) *+{y}  = "y"
\ar@{-} "r";"a"
\ar@{-} "a";"x"
\ar@{-} "a";"y"
\endxy
$$
in $F_\lambda' \{\,y\,\}$ and the function $j: \{\,y\,\} \to \Gamma$
with $x \in \Gamma$ and $j(y) = x$. Then to define the action of
$F_\lambda'$ on $j$ we must rename the bound variable $x$ to some $z
\not\in \Gamma$. But in fact, any other renaming to $z' \not\in
\Gamma$ is fine, too. So trying to define the action of $F_\lambda'$
on functions naturally forces us to consider equivalence classes
modulo $\alpha$-conversion.)
\endchg

\item
The \emph{presheaf $F_{\lambda,\Sig}$ of finite
$\lambda$-$\Sigma$-terms} is
defined analogously: in \eqref{add} we just add the term
$\sigma(\tau_1,\dots,\tau_n)$ for every $n$-ary operation symbol
$\sigma\in\Sig$, and
in~\eqref{rdj} the corresponding tree.

\item

The presheaf $T_\lambda$ of all (finite and infinite)
\emph{$\lambda$-terms} is defined analogously to $F_\lambda$. We first
denote by $T'_\lambda(\Gamma)$ the set of all trees~\eqref{rdj}
dropping the assumption of finiteness. Then we use
$\alpha$-conversion:
for infinite trees $t$ and $t'$ we write
\[t \sim_\alpha t'\]
if their
(finite) cuttings at level $k$ (with label $\bot$ for all
leaves at level $k$) are $\alpha$-equivalent in the above sense for
all $k \in \nn$.
(We can formalize this by using
$\Sigma_\bot =\Sigma\cup \{\,\bot\,\}$ with $\bot$ a
constant symbol outside of
$\Sig\cup\Var$).
The
presheaf $T_\lambda$ is defined on objects $\Gamma$ by
$T_\lambda(\Gamma) = T_\lambda'(\Gamma)/\mathord{\sim_\alpha}$ and on
morphisms $\gamma: \Gamma \to \Gamma'$ by relabellings of variables as
in~(iii). Observe that since $\Var\setminus\Gamma$ is infinite, the
relabelling of bound variables needed here causes no problem.

\item
The presheaf~$R_\lambda$ of \emph{rational $\lambda$-terms} is
also defined analogously. Recall that a tree is called
\textit{rational} if it
has up to isomorphism only finitely many subtrees. We denote by
$R_\lambda'(\Gamma)$ the set of all rational trees in
$T_\lambda'(\Gamma)$ and define a presheaf $R_\lambda$ by
$R_\lambda(\Gamma) = R_\lambda'(\Gamma)/\mathord{\sim_\alpha}$ on
objects, and by relabellings of variables (as in~(iii)) on morphisms.
Observe that, by definition, every rational $\lambda$-term $t$ is
represented by a rational $\lambda$-tree. However, $t$ can also be
represented by non-rational $\lambda$-trees---for example, if it
contains infinitely many $\lambda$'s, the $\alpha$-conversion can
introduce an infinite number of bound variables.

\item
The presheaves $T_{\lambda,\Sigma}$ (of all $\lambda$-\Sig-terms) and
$R_{\lambda,\Sig}$ (of rational $\lambda$-\Sig-terms) are obvious
modifications of~(iv) and~(v): one adds to
\eqref{add} and~\eqref{rdj} the case
$\sigma(\tau_1, \ldots, \tau_n)$ for all $n$-ary symbols $\sigma \in
\Sig$ and all (rational) $\lambda$-\Sig-trees $\tau_1, \ldots,
\tau_n$.
\end{enumerate}
\end{exa}

\begin{notation}
\label{dt}
We denote by $\delta\colon\SetF\to\SetF$ the endofunctor defined by
\[\delta X(\Gamma)=X(\Gamma+1).\]
Observe that $\delta$~preserves limits and colimits.

Note that an algebra for~$\delta$
is a
presheaf~$Y$ together with an operation $Y(\Gamma+1)\to Y(\Gamma)$
for all contexts~$\Gamma$---this is precisely the form of
$\lambda$-abstraction, where to a formula $f$ in~$Y(\Gamma+\{y\})$ we
assign $\lambda y\tecd f$ in~$Y(\Gamma)$. The other
$\lambda$-operation, application, is simply a presheaf
morphism $X\times X\to
X$, that is, a binary operation on~$X$. We put
these two together:
\end{notation}

\begin{notation}
\label{dc}
Let
$H_\lambda$~denote the endofunctor of~\SetF\ given by
\[H_\lambda X=X\times X+\delta X.\]
Thus, an algebra for~$H_\lambda$ is a presheaf~$X$ together with
operations of application $X(\Gamma)\times X(\Gamma)\to
X(\Gamma)$ and abstraction $X(\Gamma+1)\to X(\Gamma)$ for all
contexts~$\Gamma$; these operations are
compatible with the renaming of free variables.
\end{notation}

\begin{exa}
\label{dp}
The presheaves
$F_\lambda$, $T_\lambda$ and~$R_\lambda$ are algebras for~$H_\lambda$
in the obvious sense.
\end{exa}

\begin{rem}
\label{ds}\hfill
\begin{enumerate}[(i)]
\item
The slice category~$V/\SetF$ of presheaves~$X$ together with a
morphism
$i\colon V\to X$ is called the category of \textit{pointed
presheaves}. For example $F_\lambda$ is a
pointed presheaf in a canonical sense: $i^F\colon V\to F_\lambda$
takes a variable~$x$ to the term~$x$. Analogously
$i^T\colon V\to T_\lambda$
and $i^R\colon V\to R_\lambda$ are pointed presheaves, and so are
$F_{\lambda,\Sig}$, $R_{\lambda,\Sig}$ and $T_{\lambda,\Sig}$.

\item
Recall that the category~$\Alg H_\lambda$ of algebras for~$H_\lambda$
has as morphisms the usual $H_\lambda$-homomorphisms, i.e., a morphism
from $a\colon H_\lambda X\to X$ to $b\colon H_\lambda Y\to Y$ is a
natural transformation
$f\colon X\to Y$ such that $f\tec a=b\tec H_\lambda f$. Then
$\Alg H_\lambda$~is a concrete category over~\SetF\ with the forgetful
functor $(H_\lambda
X\to
X)\mapsto X$.
\end{enumerate}
\end{rem}

\begin{thm}[see \cite{FPT}]
\label{dss}
The presheaf~$F_\lambda$ of finite $\lambda$-terms is the  free
$H_\lambda$-algebra on~$V$.
\end{thm}

\begin{defi}[see \cite{AMV1}]
\label{do}
Given an endofunctor~$H$, an algebra $a\colon HA\to A$ is called

\begin{enumerate}[(1)]
\item
\textit{completely iterative} (cia for short)
if for every object~$X$ (of variables)
and
every (flat equation) morphism
$e\colon X\to HX+A$ there exists a unique \textit{solution} which
means a unique morphism
$e^\dag\colon X\to A$ such that the square below commutes
\begin{equation}
\label{rdd}
\vcenter{
\xymatrix@C+1pc@R-12pt{
  X
  \ar[r]^{e^\dag}
  \ar[d]_e
  &
  A
  \\
  HX + A
  \ar[r]_{He^\dag + \id}
  &
  HA + A
  \ar[u]_{[a, \id]}
}}
\end{equation}

\item
\textit{iterative} if every equation morphism $e\colon X\to HX+A$
with $X$~finitely presentable has a unique solution $e^\dag\colon
X\to A$.
\end{enumerate}
\end{defi}\smallskip

\noindent We are going to characterize finitely presentable presheaves in
Theorem~\ref{hdjj}. In practice, we are interested only in equations
using free presheaves (on polynomial endofunctors of~\Set) as~$X$, but
including the more general concept does not ``disturb'' anything as
we explain in Remark~\ref{td}.

\begin{exa}
\label{ndoa}
As proved in~\cite{nn}, Corollary~6.3, the free completely iterative
algebra for an arbitrary finitary
endofunctor~$H$ on an object~$X$ is precisely the
terminal coalgebra for $H({-})+X$. More detailed, suppose $TX$~is the
terminal coalgebra for $H({-})+X$, then its structure morphism is an
isomorphism by Lambek's Lemma and the inverse of this morphism has
the components
\[\tau^{\phantom{T}}_X\colon HTX\to TX\qquad\text{and}\qquad
\eta^T_X\colon X\to
TX\]
making~$TX$ a free cia on~$X$.

Conversely, let $\tau^{\phantom{T}}_X\colon HTX\to TX$ be a cia which
is free
on~$X$ w.r.t. the universal arrow~$\eta^T_X$. Then
$[\tau^{\phantom{T}}_X,\eta^T_X]\colon HTX+X\to TX$ is an isomorphism,
and its
inverse is the structure of the terminal coalgebra for $H({-})+X$.
\end{exa}

\begin{thm}
\label{ddv}
The presheaf~$T_\lambda$ of infinite $\lambda$-terms is the free
completely iterative $H_\lambda$-algebra on~$V$.
\end{thm}

\begin{proof}
As explained in Example~\ref{ndoa} above,
the free completely iterative algebra
for~$H_\lambda$ on~$V$ is precisely the terminal coalgebra for
$H_\lambda({-})+V$. The latter functor clearly preserves limits of
$\omega^{\op}$-chains. Consequently, its terminal coalgebra is a limit
of
the chain~$W$ with $W_0=1$ (the terminal presheaf) and
$W_{n+1}=H_{\lambda}W_n+V$, where the connecting maps are the unique
$w_0\colon
W_1\to W_0$ and $w_{n+1}=H_{\lambda}w_n+\id_V$.

\chg Observe first that the limit of $W$ is computed objectwise. So
for every context~$\Gamma$ we can identify~$W_0(\Gamma)$ with the
set~$\{\perp\}$ where ${\perp}\notin\Var$, and we have
$$
W_{n+1}(\Gamma) = W_n(\Gamma) \times W_n(\Gamma) + W_n(\Gamma + 1) +
\Gamma.
$$
An easy induction proof now shows that \endchg $W_n(\Gamma)$~can be
identified with the set of all $\lambda$-terms in context $\Gamma$ of
depth at most~$n$
having all leaves of depth~$n$ labelled by~$\perp$. And
$w_{n+1}\colon W_{n+1}\to W_n$
cuts
away the level~${n+1}$ in the trees of~$W_{n+1}(\Gamma)$, relabelling
level-$n$ leaves by~$\perp$.  With this identification we
obtain~$T_\lambda$ as a limit of~$W_n$ where the limit maps
$T_\lambda\to W_n$ cut the trees in~$T_\lambda(\Gamma)$ at level~$n$
and
relabel level-$n$ leaves by~$\perp$.
\end{proof}

\begin{exa}
\label{nedv}
The complete iterativity of the algebra~$T_\lambda$ means that we are
able to solve systems of recursive equations such as
\begin{equation}
  \label{eq:system}
  \begin{array}{rcl}
    p_1&=&p_1\zav(\lambda x\tecd p_2)\\
    p_2&=&y\zav p_1.
  \end{array}
\end{equation}
Indeed, the solution in~$T_\lambda(\{y\})$ is formed by the
$\lambda$-terms
represented by the following trees $\hat{t}_1$ and~$\hat{t}_2$:
\[
\hat{t}_1=\raisebox{25mm}{$
\pstree[nodesep=3pt,levelsep=10mm,treesep=12mm]{\TR{\zav}}{%
\pstree{\TR{\zav}}{%
\pstree{\TR{\zav}}{%
\TR{.\raisebox{4pt}{.}\raisebox{8pt}{.}}
\pstree{\TR{\lambda x}}{%
\pstree{\TR{\zav}}{%
\TR{y}{%
\pstree{\TR{\zav}}{%
\pstree{\TR{\zav}}{%
\pstree{\TR{\zav}}{
\pstree{\TR{\raisebox{-5pt}{\vdots}}}{}
}
}
\TR{\raisebox{8pt}{.}\raisebox{4pt}{.}.}
}
}
}
}
}
\pstree{\TR{\lambda x}}{%
\pstree{\TR{\zav}}{%
\TR{y}{%
\pstree{\TR{\zav}}{%
\pstree{\TR{\zav}}{%
\pstree{\TR{\zav}}{
\pstree{\TR{\raisebox{-5pt}{\vdots}}}{}
}
}
\TR{\raisebox{8pt}{.}\raisebox{4pt}{.}.}
}
}
}
}
}
\pstree{\TR{\lambda x}}{%
\pstree{\TR{\zav}}{%
\TR{y}{%
\pstree{\TR{\zav}}{%
\pstree{\TR{\zav}}{%
\pstree{\TR{\zav}}{
\pstree{\TR{\raisebox{-5pt}{\vdots}}}{}
}
}
\TR{\raisebox{8pt}{.}\raisebox{4pt}{.}.}
}
}
}
}
}
$}
\quad
\hat{t}_2=\quad
\raisebox{10mm}{$
\pstree[levelsep=10mm,treesep=10mm]{%
\TR{\raisebox{-3pt}{\rule{0pt}{12pt}}\zav}}{%
\TR{\raisebox{-8pt}{$y$}\rule{8pt}{0pt}}
\pstree{\Tp~[tnpos=b,tnsep=5mm]{\hat{t}_1}}{\Tfan}}$}
\]\vspace{-12 pt}

\noindent How is this related to the above concept of Definition~\ref{do}?
Firstly, every system of recursive equations can be flattened: a flat
system has in context~$\Gamma$ the right-hand sides of only three
types:
$p_i\zav p_j$ or $\lambda x\tecd p_i$ or a term
in~$T_\lambda(\Gamma)$. For example, we flatten the
system~\refeq{eq:system} to
\begin{equation}
  \label{eq:sysflat}
  \begin{array}{rcl}
    p_1& = & p_1\zav p_3\\
    p_2& = & p_4\zav p_1\\
    p_3& = & \lambda x\tecd p_2\\
    p_4& = & y
  \end{array}
\end{equation}
Let $\Gamma=\{y\}$ be the context of all free variables and let
$X$~be the free presheaf on generators $p_1,\dots,p_4$ of
context~$\Gamma$, see Example~\ref{dd}. \chg
Notice that even though
the recursion variables
$p_1, \ldots, p_4$ appear as constants in the
system~\refeq{eq:sysflat}, the associated  presheaf $X$ is not a
constant presheaf. Using the Yoneda lemma, the above
system~\refeq{eq:sysflat} \endchg
defines an obvious morphism
\[e\colon X\to H_\lambda X+T_\lambda\]
viz, the unique one such that $e_\Gamma(p_i)$~is the right-hand side
of the equation above. The solution
\[e^\dag\colon X\to T_\lambda\]
is the unique morphism such that $e^\dag_\Gamma$ takes~$p_i$ to the
solution in~$T_\lambda$; for example $e^\dag_\Gamma(p_1)=[\hat{t}_1]$
for the above tree~$\hat{t}_1$. We will see in Theorem~\ref{hncss}
below that equations such as~\eqref{eq:sysflat} have a unique
solution yielding rational trees.
\end{exa}

\begin{rem}
\label{nrmdv}
Given an equation morphism
\[e\colon X\to H_\lambda X+T_\lambda\]
then the solution $e^\dag\colon X\to T_\lambda$ allows us to choose,
for every element~$p$ of~$X(\Gamma)$, a tree~$\hat{t}_p$ in
${T}_\lambda(\Gamma)$ with
\[e^\dag_\Gamma(p)=[\hat{t}_p].\]
Due to the commutativity of~\eqref{rdd} for $H_\lambda X=X\times
X+\delta X$ we have three possible cases for every~$p$:

\begin{enumerate}[(a)]
\item
$e_\Gamma(p)=(p_1,p_2)$ in $X(\Gamma)\times X(\Gamma)$, then for the
operation $\tau\colon H_\lambda T_\lambda\to T_\lambda$ we have
\[[\hat{t}_p]=\tau\bigl([\hat{t}_{p_1}],[\hat{t}_{p_2}]\bigr)\]
in other words,
\[\hat{t}_p\sim_\alpha
\raisebox{7mm}{$
\pstree[nodesep=3pt,levelsep=15mm]{%
\TR{\zav}}{%
\TR{\hat{t}_{p_1}}
\TR{\hat{t}_{p_2}}}$}
\]

\item
$e_\Gamma(p)=q$ in $X(\Gamma+\{x\})$, then
\[[\hat{t}_p]=\tau\bigl([\hat{t}_q]\bigr)\]
in other words
\[\hat{t}_q\sim_\alpha
\raisebox{7mm}{$
\pstree[nodesep=3pt,levelsep=15mm]{%
\TR{\lambda x}}{%
\TR{\hat{t}_{p}}}$}
\]
or

\item
$e_\Gamma(p)$~lies in~$T_\lambda(\Gamma)$ and is represented
by~$\hat{t}_p$:
\[e^\dag_\Gamma(p)=[\hat{t}_p]=e_\Gamma(p).\]
\end{enumerate}
\end{rem}

\begin{rem}
\label{djn}
We are going to characterize the presheaf~$R_\lambda$ as a free
iterative algebra for~$H_\lambda$.
That is, in equations we admit
only presheaves~$X$ of variables
that are \textit{finitely presentable}. Recall that an object~$X$ of
a category~$\WW$ is \textit{finitely presentable} provided that its
$\hom$-functor $\WW(X,{-})$
preserves filtered colimits.
We are first going to
characterize the finitely presentable presheaves by
using the following concept:
\end{rem}

\begin{defi}[see \cite{AMV2}]
\label{ndfdjj}
A presheaf~$X$ is called
\textit{super-finitary} provided that each~$X(\Gamma)$ is finite and
there exists a nonempty context~$\Gamma_0$
\textit{generating~$X$} in the sense
that for every nonempty context~$\Gamma$ we have
\begin{equation}
\label{rdt}
X(\Gamma)=\bigcup_{\gamma\colon\Gamma_0\to\Gamma}
X\gamma\bigl[X(\Gamma_0)\bigr].
\end{equation}
\end{defi}

\begin{exa}
\label{plus}
A signature~\Sig\ defines the polynomial presheaf~$X_{\Sig}$,
see Example~\ref{dd}(ii), by
$X_{\Sig}(\Gamma)=\coprod_{\sigma\in\Sig}
\Gamma^{\mathrm{ar}(\sigma)}$. This
is a super-finitary presheaf iff \Sig~is a finite signature. Other
super-finitary presheaves are precisely the quotients
of~$X_{\Sig}$ with \Sig~finite.
\end{exa}

\begin{thm}
\label{hdjj}
A presheaf in~\SetF\ is finitely presentable iff it is
super-fi\-ni\-ta\-ry.
\end{thm}

\proof
(1)\ Let $X$~be a super-finitary presheaf and let $\Gamma_0$~be a context
of $n$~variables generating~$X$. We prove that $X$~is a finite
colimit of representables. Since representables are (due to Yoneda
lemma) clearly finitely presentable, this proves finite presentability
of~$X$.

Form the finite diagram of all presheaves
\[Z_a=\FF(\Gamma,{-})\]
where $\Gamma\subseteq\Gamma_0+\Gamma_0$
is a context of at most $2n$~variables\footnote{The reason why we
need $2n$ variables will become clear in~\eqref{nnrdss} below.}
and $a\in
X(\Gamma)$. The connecting morphisms are the Yoneda transformations
\[Yf\colon Z_a\to Z_{a'}\qquad\text{for $a\in X(\Gamma)$ and $a'\in
X(\Gamma')$}\]
where $f\colon\Gamma'\to\Gamma$ is a function that
fulfils $Xf(a')=a$. The
Yoneda transformations
\[z_a\colon Z_a\to X,\quad\text{with the components defined by}\quad
f\mapsto Xf(a),\]
clearly form a compatible cocone of this finite diagram. We prove
that this is a colimit cocone. In other words, for every
context~$\bar{\Gamma}$ we must prove that the cocone of all
$\bar{\Gamma}$-components~$z_a^{\bar{\Gamma}}$ (sending elements
$f\colon\Gamma\to\bar{\Gamma}$ of $Z_a=\FF(\Gamma,{-})$ to~$Xf(a)$)
is a colimit in~$\Set$. For that we only need to verify that
in every
context~$\bar{\Gamma}$
\begin{enumerate}[(i)]
\item
the cocone~$z_a^{\bar{\Gamma}}$
is collectively epimorphic, and
\item
whenever two elements $f\colon\Gamma\to\bar{\Gamma}$ of~$Z_a$ and
$f'\colon\Gamma'\to\bar{\Gamma}$ of~$Z_{a'}$ fulfil
$z_a^{\bar{\Gamma}}(f)=z_{a'}^{\bar{\Gamma}}(f')$, then there exists a
zig-zag connecting $f$ and~$f'$ in the $\bar{\Gamma}$-component of
our diagram.
\end{enumerate}
The proof of~(i) is trivial: given an element $a\in X(\bar{\Gamma})$,
either $\bar{\Gamma}=\emptyset$ or by Equation~\eqref{rdt}
there exists $f\colon\Gamma_0\to\bar{\Gamma}$ and an element $b\in
X(\Gamma_0)$ with $a=Xf(b)$, in other words,
\[a=z_b^{\bar{\Gamma}}(f).\]
In case $\bar{\Gamma}=\emptyset$ we have
$a=z_a^{\bar{\Gamma}}(\id_\emptyset)$.

To prove~(ii), observe that the given equation states
\[Xf(a)=Xf'(a').\]
In case $\bar{\Gamma}$~has at most $2n$~variables, we can assume
$\bar{\Gamma}\subseteq\Gamma_0+\Gamma_0$ and
the desired zig-zag
is
\[Z_a\lnsi{Yf}Z_b\nsi{Yf'}Z_{a'},\]
where $b=Xf(a)$.
Thus, we can assume that
$\bar{\Gamma}$~has more than $2n$~elements.

\begin{enumerate}[\hbox to8 pt{\hfill}]
\item\noindent{\hskip-12 pt\bf Case~1:}\
$\Gamma=\emptyset=\Gamma'$. Here $f=f'$ and we have
$Xf(a)=Xf(a')$. Choose a monomorphism
$m\colon\Gamma_0\to\bar{\Gamma}$ and observe that
$f=m\tec g$ for the unique
$g\colon\emptyset\to\Gamma_0$.
Thus $Xm(Xg(a))=Xm(Xg(a'))$ and since $m$~is a split monomorphism, we
conclude $Xg(a)=Xg(a')=c$. The desired zig-zag is
\[Z_a\lnsi{Yg}Z_c\nsi{Yg'}Z_{a'}.\]

\item\noindent{\hskip-12 pt\bf Case~2:}\
$\Gamma=\emptyset\neq\Gamma'$. Factorize~$f'$ as an epimorphism~$e$
followed by a split monomorphism~$m$:
\[\begin{psmatrix}[colsep=15mm, rowsep=10mm]
\Gamma'&&\bar{\Gamma}\\
&\Gamma_1
\psset{arrows=->,nodesep=3pt}
\everypsbox{\scriptstyle}
\ncline{1,1}{1,3}^{f'}
\ncline{->>}{1,1}{2,2}<{e}
\ncline{>->}{2,2}{1,3}>{m}
\end{psmatrix}\]
Then, since for the unique $h\colon\emptyset\to\Gamma_1$ we have
$f=m\tec h$, we obtain
\[Xm\bigl(Xe(a')\bigr)=Xm\bigl(X h(a)\bigr).\]
Thus, $Xe(a')=Xh(a)=c$ which yields the zig-zag
\[Z_a\lnsi{Yh}Z_c\nsi{Ye}Z_{a'}.\]

\item\noindent{\hskip-12 pt\bf Case~3:}\
$\Gamma\neq\emptyset\neq\Gamma'$. Find
$g\colon\Gamma_0\to\Gamma$ with $a=Xg(b)$ and
$g'\colon\Gamma_0\to\Gamma'$ with $a'=Xg'(b')$
for some $b,b'\in X(\Gamma_0)$. Then $X(f\tec g)(b)=X(f'\tec
g')(b')$. Now factorize $[f\tec g,f'\tec g']\colon
\Gamma_0+\Gamma_0\to\bar{\Gamma}$ as an epimorphism
followed by a split monomorphism; so we obtain a
commutative diagram

\begin{equation}
\label{nnrdss}
\myvc{
\begin{psmatrix}[colsep=15mm,rowsep=10mm]
\Gamma_0+\Gamma_0&&\bar{\Gamma}\\
&\Gamma_1
\psset{arrows=->,nodesep=3pt}
\everypsbox{\scriptstyle}
\ncline{1,1}{1,3}\naput{[f\tec g,f'\tec g']}
\ncline{->>}{1,1}{2,2}\nbput[npos=.333]{[e,e']}
\ncline{>->}{2,2}{1,3}\nbput[npos=.5]{m}
\end{psmatrix}
}
\end{equation}
Since $m$~is a split monomorphism, conclude that $Xe(b)=Xe'(b')=c$.
The desired zig-zag is

\[\begin{psmatrix}[colsep=15mm,rowsep=10mm]
Z_a&&Z_{a'}\\
Z_b&&Z_{b'}\\
&Z_c
\psset{arrows=->,nodesep=3pt}
\everypsbox{\scriptstyle}
\ncline{1,1}{2,1}<{Yg}
\ncline{1,3}{2,3}>{Yg'}
\ncline{3,2}{2,1}\naput[npos=.5]{Ye}
\ncline{3,2}{2,3}\nbput[npos=.5]{Ye'}
\end{psmatrix}\]
\end{enumerate}\smallskip

\noindent (2)\ Let $X$~be a finitely presentable object of~\SetF. The empty maps are
denoted by $t_\Gamma\colon\emptyset\to\Gamma$. For every nonempty
context~$\Gamma_0$ let $X_{\Gamma_0}$~be the subfunctor of~$X$
generated by the elements of $X(\Gamma_0)\cup X(\emptyset)$: it
assigns to every~$\Gamma$ the subset of~$X(\Gamma)$ given by
\[X_{\Gamma_0}(\Gamma)=Xt_\Gamma\bigl[X(\emptyset)\bigr]
\cup\bigcup_{f\colon\Gamma_0\to\Gamma} Xf\bigl[X(\Gamma_0)\bigr].\]
We obviously have a union
\[X=\bigcup_{\Gamma_0\in\FF\setminus\{\emptyset\}}X_{\Gamma_0}\]
which is directed: given nonempty contexts~$\Gamma_0,\Gamma_1$ then
$X_{\Gamma_0}\cup X_{\Gamma_1}\subseteq X_{\Gamma_0\cup\Gamma_1}$.
Since $X$~is finitely presentable, the morphism
\[\id_X\colon X\to\colim_{\Gamma_0\in\FF\setminus\{\emptyset\}}
X_{\Gamma_0}\]
factorizes through one of the colimit injections
$X_{\Gamma_0}\hookrightarrow X$. In other words
\[X=X_{\Gamma_0}\qquad\text{for some $\Gamma_0\neq\emptyset$.}\]
It remains to prove that the sets $X(\Gamma_0)$ and~$X(\emptyset)$
are finite, then every~$X(\Gamma)$ is finite.

For every finite set $M\subseteq X(\emptyset)$ we have the
subfunctor~$X^M$ of~$X$ equal to~$X$ on nonempty objects and maps,
and assigning~$M$ to~$\emptyset$. We obviously get~$X$ as a directed
union of these subfunctors~$X^M$, thus, as above, there exists~$M$
with $X=X^M$. Then $X(\emptyset)=M$ is finite.

For every finite set $M\subseteq X(\Gamma_0)$ we have the
subfunctor~$^M\!X$ of $X=X_{\Gamma_0}$ generated by the elements of
$M\cup X(\emptyset)$:
\[^M\!X(\Gamma)=Xt_\Gamma\bigl[X(\emptyset)\bigr]\cup
\bigcup_{f\colon\Gamma_0\to\Gamma} Xf[M].\]
Again $X$~is a directed union of these subfunctors~$^M\!X$, thus,
there exists~$M$ with $X={}^M\!X$, proving that $X(\Gamma_0)$ is
finite.\qed


\begin{thm}
\label{hdjt}
The presheaf~$R_\lambda$ of rational $\lambda$-terms is the free
iterative $H_\lambda$-algebra on~$V$.
\end{thm}

\proof
(I)\ 
$R_\lambda$~is an iterative algebra for~$H_\lambda$. Indeed, given
an equation morphism
\[e\colon X\to H_\lambda X+R_\lambda\]
where Equation~\eqref{rdt} holds for~$\Gamma_0$,
we know that its extension
\[
\bar{e}\colon X\nsi{e}
H_\lambda X+R_\lambda\hookrightarrow H_\lambda X+T_\lambda
\]
has a unique solution $e^\dag\colon X\to T_\lambda$, and we are going
to prove that the trees $e^\dag_{\Gamma_0}(p)$
and~$e^\dag_\emptyset(p)$
are all
rational. It
then follows that all the trees~$e^\dag_\Gamma(p)$
are rational for all contexts~$\Gamma$,
and this gives us the desired solution
$X\to R_\lambda$.
Indeed, for each $x\in X(\Gamma)$ with $\Gamma\neq\emptyset$
we have $x=Xf(p)$ for some
$f\colon\Gamma_0\to\Gamma$ and $p\in X(\Gamma_0)$. Then
$e^\dag_\Gamma(x)=e^\dag_\Gamma(Xf(p))=T_\lambda
f(e^\dag_{\Gamma_0}(p))$ by the naturality of~$e^\dag$, and since
$e^\dag_{\Gamma_0}(p)$~is rational, so is~$T_\lambda
f(e^\dag_{\Gamma_0}(p))$. (The action of~$T_\lambda f$ is just
relabelling leaves according to~$f$.)

Now every element of $X(\Gamma_0)=\{p_1,\dots,p_n\}$
yields an element
\[e_{\Gamma_0}(p_i)\in X(\Gamma_0)\times X(\Gamma_0)+
X\bigl(\Gamma_0+\{x\}\bigr)+R_\lambda(\Gamma_0)\]
which is either
(i)~a pair $(p_j,p_k)$ or
(ii) $q\in X(\Gamma_0+\{x\})$
or
(iii)~a rational tree
in~$R_\lambda(\Gamma_0)$.
Put $t_i=e^\dag_{\Gamma_0}(p_i)$, then
in the last
case the commutativity of Diagram~\eqref{rdd} implies that
$e_{\Gamma_0}(p_i)=t_i$ (cf. Remark~\ref{nrmdv}).
From~\eqref{rdd} we also obtain in cases (i) and~(ii)
\[t_i=t_j\zav t_{k}\qquad\text{and}\qquad t_i=\lambda x\tecd
e^\dag_{\Gamma_0+\{x\}}(q),\qquad\text{respectively.}\]
From
Equation~\eqref{rdt} we see that in case~(ii) there exists $f\colon
\Gamma_0\to\Gamma_0+\{x\}$ with $q=Xf(p_j)$ for some~$j$, then
$e^\dag_{\Gamma_0+\{x\}}(q)=T_\lambda f(e^\dag_{\Gamma_0}(p_j))=
T_\lambda f(t_j)$. Thus we get equations telling us that for
every~$i$ either $t_i=t_j\zav t_k$ or $t_i=\lambda x\tecd T_\lambda
f(t_j)$ or $t_i$~is a rational tree.
Using these equations it is now easy, for every $i=1,\dots,n$,
to prove by induction on the depth~$k$ of subtrees of~$t_i$ that
each
subtree of~$t_i$
is either of the
form~$s=T_\lambda f(e^\dag_{\Gamma_0}(r))$ for
some $r\in X(\Gamma_0)$ and some $f\colon\Gamma_0\to \Gamma_0+\{x\}$,
or $s$~is a subtree of some rational tree
$e^\dag_{\Gamma_0}(r)=e_{\Gamma_0}(r)$ in case~(iii). Since
$X(\Gamma_0)$~is a finite set, it follows that every tree~$t_i$ has
only finitely many subtrees, whence $t_i\in R_\lambda(\Gamma_0)$.

The case $X(\emptyset)=\{p_1,\dots,p_n\}$ is analogous: for
$t_i=e^\dag_\Gamma(p_i)$ we get
(i) $t_i=t_j\zav t_k$ or
(ii)
$t_i=\lambda x\tecd e^\dag_{\{x\}}(q)$ or
(iii) $t_i=e_\emptyset(p_i)\in R_\lambda(\emptyset)$. We already know
that the trees in case~(ii) are rational. Thus, each subtree
of~$e^\dag_\emptyset(p_i)$ is either~$e^\dag_\emptyset(r)$ or it is a
subtree of some rational tree in cases (ii) or~(iii).

The solution of~$e$
in~$R_\lambda$ is unique because every solution
in~$R_\lambda$ yields a solution of the extended morphism~$\bar{e}$
in~$T_\lambda$.\medskip

\noindent(II)\ 
Let \DD~be the category of all equation morphisms
\[e\colon X\to H_\lambda X+V,\qquad\text{$X$ finitely presentable,}\]
whose morphisms are the coalgebra homomorphisms
for~$H_\lambda({-})+V$. The diagram $D\colon\DD\to\SetF$,
$D(e)=X$, is
filtered and its colimit is the free iterative $H_\lambda$-algebra
on~$V$, see~\cite{AMV1}. We will prove that $R_\lambda$~is a colimit
of~$D$. Recall that $R_\lambda$~is a pointed presheaf (see
Remark~\ref{ds}).

For every $e$ as above the equation morphism
\[\tilde{e}\equiv X\nsi{e}H_\lambda X+V\nsi{\id+i^R}H_\lambda
X+R_\lambda\]
has a unique solution $\tilde{e}^\dag\colon X\to R_\lambda$.
It is easy to verify that these morphisms form a cocone
for
the diagram~$D$. Since $D$~is a filtered diagram in~\SetF\ and since
colimits in~\SetF\ are constructed objectwise in~\Set, in order to
prove that
\[R_\lambda=\colim D\qquad\text{with the colimit cocone
$(\tilde{e}^\dag)$}\]
all we need to prove is that for every context~$\Gamma$
\begin{enumerate}[(a)]
\item
the cocone $\tilde{e}^\dag_\Gamma$ is collectively epimorphic:
$R_\lambda(\Gamma)=\bigcup \tilde{e}^\dag_\Gamma[X]$, and
\item
whenever $\tilde{e}^\dag_\Gamma$~merges $x,x'\in X(\Gamma)$, there
exists a connecting morphism in~$\DD$ merging $x$ and~$x'$ too.
\end{enumerate}\smallskip

\noindent To prove~(a), let $t\in R_\lambda(\Gamma)$ be a rational tree and let
$\Gamma_0$~be the context of variables~$x_s$ indexed by the finitely
many
subtrees~$s$
of~$t$ (up to isomorphism). Let $X$~be the free presheaf
on the set~$\Gamma_0$ of generators of context
$\bar{\Gamma}=\Gamma\cup\Gamma_0$,
see Example~\ref{dd}(ii). Define
\[e\colon X\to H_\lambda X+V\]
by assigning to every variable~$x_s$, for a
subtree~$s$ of~$t$, the following value: if
$s=s'\zav s''$ in~$t$, then
\[e_\Gamma(x_s)=x_{s'}\zav x_{s''}\qquad\text{in
$X(\bar{\Gamma})
\times X(\bar{\Gamma})$,}\]
if $s=\lambda y\tecd s'$ in~$t$, then
\[e_\Gamma(x_s)=\lambda y{.}x_{s'}\qquad\text{in
$X\bigl(\bar{\Gamma}+
\{y\}\bigr)$,}\]
and if $s$~is a leaf labelled by $x\in\Gamma$, then
\[e_\Gamma(x_s)=x\qquad\text{in $\Gamma=V(\Gamma)$.}\]
This object~$e$ of~\DD\ yields two equation morphisms:
$\tilde{e}\colon X\to H_\lambda X+R_\lambda$ above, and analogously
$\hat{e}=(\id+i^T)\tec e\colon
X\to H_\lambda X+T_\lambda$. The solution of the
latter is the unique morphism
\[\hat{e}^\dag\colon X\to T_\lambda\qquad\text{with\qquad
$\hat{e}^\dag_{\bar{\Gamma}}(x_s)=s$ for all $s\in\Gamma_0$.}\]
Indeed,
Diagram~\eqref{rdd}~is easily seen to commute for $\hat{e}$
and~$\hat{e}^\dag$. In (I)~above we saw that the
solution $\tilde{e}^\dag\colon X\to R_\lambda$ is a codomain
restriction of~$\hat{e}^\dag$. In particular:
\[t=\tilde{e}^\dag_{\bar{\Gamma}}(x_t).\]
This proves~(a).

To prove~(b) let $\tau\colon H_\lambda T_\lambda\to T_\lambda$ denote
the algebra structure of~$T_\lambda$. By
Theorem~\ref{ddv} and Example~\ref{ndoa} we
have that
\[[\tau,i^T]\colon H_\lambda T_\lambda+V\to T_\lambda\qquad\text{is an
isomorphism.}\]
From Diagram~\eqref{rdd} we get
\[
\hat{e}^\dag=[\tau,\id_{T_\lambda}]\tec[H_\lambda\hat{e}^\dag+
\id_{T_\lambda}]\tec(\id_{H_\lambda X}+i^T)\tec e
\]
which yields
\[[\tau,i^T]^{-1}\tec\hat{e}^\dag=
(H_\lambda\hat{e}^\dag+\id_V)\tec e.\]
Let us factorize~$\hat{e}^\dag$ as a strong epimorphism $k\colon
X\to
Y$ followed by a monomorphism $m\colon Y\to T_\lambda$. Then the last
equation makes it possible to apply the diagonal fill in:
\vspace{1pt}

\[\begin{psmatrix}[rowsep=10mm]
X&Y\\
H_\lambda X+V&T_\lambda\\
H_\lambda Y+V&H_\lambda T_\lambda+V
\psset{arrows=->,nodesep=3pt}
\everypsbox{\scriptstyle}
\ncline{1,1}{1,2}^{k}
\ncline{1,2}{2,2}>{m}
\ncline{1,1}{2,1}<{e}
\ncline{2,1}{3,1}<{H_\lambda k+\id}
\ncline{2,2}{3,2}>{[\tau,i^T]^{-1}}
\ncline{3,1}{3,2}_{H_\lambda m+\id}
\ncline[linestyle=dashed]{1,2}{3,1}>f
\end{psmatrix}\]
\vspace{3pt}

\noindent
Indeed, $H_\lambda=({-})^2+\delta$ preserves connected
limits (because each summand does), thus, monomorphisms;
consequently, $H_\lambda m+\id_V$ is a monomorphism. Since $Y$~is a
strong
quotient of~$X$, it follows from Theorem~\ref{hdjj} that $Y$~is
finitely presentable. Thus,
\[f\colon Y\to H_\lambda Y+V\]
is an object of~\DD, and clearly $k$~is a connecting morphism
from~$e$ to~$f$.

From~(I) we know that $\tilde{e}^\dag$~is
the domain restriction of~$\hat{e}^\dag$, thus we see that
$\tilde{e}^\dag_\Gamma(x)=\tilde{e}^\dag_\Gamma(x')$ implies
$\hat{e}^\dag_\Gamma(x)=\hat{e}^\dag_\Gamma(x')$, and since
$m_\Gamma$~is a monomorphism with
$\hat{e}^\dag_\Gamma=m_\Gamma\tec
k_\Gamma$, we conclude
\[k_\Gamma(x)=k_\Gamma(x')\]
as requested.\qed

\chg
\begin{rem}
\label{ndjj}
As mentioned in the Introduction
we want to combine application and abstraction with other
operations. Suppose $\Sig=(\Sig_n)_{n\in\nn}$ is a signature (of
``terminals'').
Then we
can form the endofunctor~$H_{\lambda,\Sig}$ of~\SetF\ on objects by
$$
H_{\lambda,\Sig}X
=
X\times X+\delta X+\coprod_{n\in\nn}\Sig_n\bullet X^n
$$
where $\Sig_n\bullet X^n$ is the coproduct (that is: disjoint union
in every context) of $\Sig_n$~copies of
the $n$-th Cartesian power of~$X$. For this
endofunctor an algebra is an $H_\lambda$-algebra~$A$ together with an
$n$-ary operation on~$A(\Gamma)$ for every $\sigma\in\Sig_n$ and every
context~$\Gamma$.
\end{rem}
\endchg

In the following result we use notation of Example~\ref{dd}(vii).

\begin{thm}
\label{djp}
For every signature~\Sig
\begin{enumerate}[\em(i)]
\item
$F_{\lambda,\Sig}$~is the free $H_{\lambda,\Sig}$-algebra on~$V$,
\item
$R_{\lambda,\Sig}$~is the free iterative $H_{\lambda,\Sig}$-algebra
on~$V$,
and
\item
$T_{\lambda,\Sig}$~is the free completely iterative
$H_{\lambda,\Sig}$-algebra on~$V$.
\end{enumerate}
\end{thm}\smallskip

\noindent Indeed, (i)~was proved in~\cite{FPT}, and the proofs of (ii)
and~(iii) are completely analogous to the proofs of Theorems
\ref{hdjt} and~\ref{ddv}.

\section{Presheaves as Monoids}
\setcounter{equation}{0}
\label{ht}

\noindent So far we have not treated one of the basic features of
$\lambda$-calculus: substitution of subterms. For the
presheaf~$F_{\lambda,\Sig}$ of finite $\lambda$-\Sig-terms this was
elegantly performed by Fiore \textit{et al}~\cite{FPT}
based on the monoidal
structure of the category~\SetF. As mentioned in
Notation~\ref{dj}(3),
we can work
with the equivalent category~$\Fin(\Set,\Set)$ of all finitary
endofunctors of~\Set. Composition of functors makes this a
(strict, non-symmetric) monoidal category with unit~$\Id_{\Set}$.
This monoidal structure,
as shown in~\cite{FPT},
corresponds to simultaneous substitution. Indeed, let $X$ and~$Y$ be
objects of~$\Fin(\Set,\Set)$. Then the ``formulas of
the composite presheaf~$X\tec Y$'' in context~$\Gamma$
are the elements of
\begin{equation}
\label{newr}
X\tec Y(\Gamma)=X\bigl(Y(\Gamma)\bigr)=
\bigcup_{u: \bar{\Gamma}\subto Y(\Gamma)}Xu[\bar{\Gamma}],
\end{equation}
where $u: \bar\Gamma \subto Y(\Gamma)$ ranges over finite subobjects
of
$Y(\Gamma)$. Indeed, $X$~preserves the filtered colimit $Y(\Gamma) =
\colim \bar\Gamma$.

Consequently, in order to specify an $X\tec Y$-formula~$t$ in
context~$\Gamma$ we need (a)~an $X$-formula~$s$ in some new
context~$\bar{\Gamma}$ and (b)~for every variable $x\in\bar{\Gamma}$
a $Y$-formula of context~$\Gamma$, say, $r_x$. We can then think
of~$t$ as the formula~$s(r_x/x)$ obtained from~$s$ by simultaneous
substitution.

\chg
\begin{exa}
  We consider the presheaves $F_{\lambda,\Sigma}$ and
  $F_{\lambda,\Sigma'}$, where $\Sigma$ is the signature with a binary
  operation symbol $*$ and $\Sigma'$ a signature with a unary
  operation symbol $o$. Then for every context $\Gamma$, the elements
  of $F_{\lambda,\Sigma} \tec F_{\lambda,\Sigma'} (\Gamma)$ are
  $\lambda$-$\Sig$-terms in some context $\ol\Gamma$ with free
variables replaced by
  $\lambda$-$\Sig'$-terms in context $\Gamma$. For a concrete example,
  let $\ol\Gamma = \{\,y,z\,\}$ and $\Gamma = \{\,z'\,\}$ and consider
  the $\lambda$-$\Sig$-term
  $$
  t = \lambda x. x * (y * z)
  \qquad
  \text{in $F_{\lambda,\Sigma}\{\,y,z\,\}$}
  $$
  and the function
  $$
  u: \ol\Gamma \to F_{\lambda,\Sig'} (\Gamma)
  \qquad \text{with}\
  \begin{array}{rcl}
    u(y) & = & \lambda x. o(x) \zav z' \\
    u(z) & = &  z'\zav o(o(z'))
  \end{array}
  $$
  Then the element of $F_{\lambda,\Sig} \tec F_{\lambda,\Sig'}(\Gamma)$
  corresponding to $t$ and $u$ is the term
  $$
  \lambda x.x * \left((\lambda x. o(x) \zav z') * (z' \zav o(o(z')))\right).
  $$
\end{exa}
\endchg

\begin{rem}
\label{hndj}\hfill
\begin{enumerate}[(i)]
\item
The monoidal structure on~\SetF\ corresponding to composition
in 
$\Fin(\Set,\Set)$ will be denoted by~$\otimes$. Its unit
(corresponding to~$\Id$) is~$V$,
see Notation~\ref{dd}(i).
Observe that every
endofunctor~${-}{\otimes}X$ preserves colimits, e.g., $(A+B)\otimes
X\cong(A\otimes X)+(B\otimes X)$.

\item
Explicitly, the monoidal structure can be described by the coend
\begin{equation}
\label{eq:3.2}
(X\otimes Y)(\Gamma)
=
\int^{\overline{\Gamma}}
\Set(\overline{\Gamma},Y(\Gamma))\bullet X(\overline{\Gamma}).
\end{equation}

\item
Recall that monoids in
the monoidal category~$\Fin(\Set,\Set)$
are precisely the finitary
monads on~\Set.

\item
The presheaf~$F_{\lambda,\Sig}$ is endowed with
the usual simultaneous
substitution of $\lambda$-terms which defines a morphism $m^F\colon
F_{\lambda,\Sig}\otimes F_{\lambda,\Sig}\to F_{\lambda,\Sig}$.
Together
with the canonical pointing
$i^F\colon V\to F_{\lambda,\Sig}$, see Remark~\ref{ds},
this constitutes a monoid
as proved in~\cite{FPT}.

Analogously the simultaneous substitution of infinite
$\lambda$-terms defines a
monoid
\[(T_{\lambda,\Sig},m^T,i^T).\]
It is easy to see that given a rational term,
every simultaneous substitution of rational terms for variables
yields again a rational term. Thus, we have a submonoid
$(R_{\lambda,\Sig},m^R,i^R)$.

\item
The monoidal operation of~$F_{\lambda,\Sig}$ is well
connected
to its structure of an $H_{\lambda,\Sig}$-algebra. This was
expressed in~\cite{FPT} by the concept of an
\textit{$H_{\lambda,\Sig}$-monoid}.

In order to recall this concept, we need the notion of
point-strength introduced in~\cite{F} under the name
$(I/\WW)$-strength; this is a weakening of the
classical strength (necessary since $H_{\lambda,\Sig}$~is
unfortunately not strong). Recall that given an object~$I$ of a
category~\WW, then
objects of the slice
category~$I/\WW$ are morphisms $x\colon I\to X$ for $X\in\obj\WW$.
\end{enumerate}
\end{rem}

\begin{defi}[see~\cite{F}]
\label{hntj}
Let $(\WW,{\otimes},I)$~be a strict monoidal category and $H$ an
endofunctor on~$\WW$. A \textbf{point-strength} of~$H$ is a collection
of morphisms
\[s_{(X,x)(Y,y)}\colon
HX\otimes Y\to H(X\otimes Y)\]
natural in $(X,x)$ and~$(Y,y)$ ranging through~$I/\WW$ such that
\begin{enumerate}[(i)]
\item
$s_{(X,x)(I,\id)}=\id_{HX}$, and
\item
the following triangles commute:\smallskip

\begin{equation}
\label{triangle}
\myvc{
\begin{psmatrix}[colsep=10mm]
HX\otimes Y\otimes Z&&H(X\otimes Y\otimes Z)\\
&H(X\otimes Y)\otimes Z
\psset{arrows=->,nodesep=3pt}
\everypsbox{\scriptstyle}
\ncline{1,1}{1,3}^{s_{(X,x),(Y\otimes Z,y\otimes z)}}
\ncline{1,1}{2,2}<{s_{(X,x),(Y,y)}\otimes\id_Z}
\ncline{2,2}{1,3}>{s_{(X\otimes Y,x\otimes y),(Z,z)}}
\end{psmatrix}
}
\end{equation}\smallskip
\end{enumerate}
\end{defi}

\begin{exa}
\label{hntd}\hfill
\begin{enumerate}[(i)]
\item The endofunctor $X\mapsto X\otimes X$ (which usually fails to be
strong) has the point-strength
\[\qquad\enspace s_{(X,x)(Y,y)}=(X\otimes X)\otimes Y=
(X\otimes I\otimes X)\otimes Y\nsi{\id_X\otimes y\otimes
\id_{X\otimes Y}}(X\otimes Y)\otimes(X\otimes Y).\]

\item
The endofunctor $X\mapsto X^n$ of~\SetF\ is clearly
(point-)strong for every $n \in \nn$.

\item
The functor
$\delta$ in Notation~\ref{dt}
is point-strong, as observed in~\cite{FPT}.
The easiest way to describe its point-strength is by working in
$\Fin(\Set,\Set)$. Given pointed endofunctors $x\colon\Id\to X$ and
$y\colon\Id\to Y$,
then the point-strength $s_{(X,x)(Y,y)}\colon
(\delta
X)\tec Y\to\delta(X\tec Y)$ has components
\[X\bigl(Y(\Gamma)+1\bigr)\nsi{X(\id+y_1)}
X\bigl(Y(\Gamma)+Y(1)\bigr) \nsi{X\can} X\tec Y(\Gamma+1),\]
where
$\can\colon Y(\Gamma)+Y(1)\to Y(\Gamma+1)$ denotes the canonical
morphism.

\item
A coproduct of point-strong functors is point-strong.
\end{enumerate}
\end{exa}

\begin{cor}
\label{hntt}
The endofunctors $H_\lambda$ and~$H_{\lambda,\Sig}$ are point-strong.
Their point-strength is denoted by~$s^H$.\qed
\end{cor}

\begin{defi}[see~\cite{FPT}]
\label{hntc}
Let $H$~be a point-strong endofunctor of a monoidal category. By an
\textbf{$H$-monoid} is meant an $H$-algebra~$(A,a)$ which is also a
monoid
\[m\colon A\otimes A\to A\qquad\text{and}\qquad i\colon I\to A\]
such that the square below commutes:\smallskip

\begin{equation}
\label{abb}
\myvc{
\begin{psmatrix}[colsep=15mm,rowsep=10 mm]
HA\otimes A&H(A\otimes A)&HA\\
A\otimes A&&A
\psset{arrows=->,nodesep=3pt}
\everypsbox{\scriptstyle}
\ncline{1,1}{1,2}\naput{s_{(A,i)(A,i)}}
\ncline{1,1}{2,1}<{a\otimes\id}
\ncline{2,1}{2,3}_{m}
\ncline{1,2}{1,3}\naput{Hm}
\ncline{1,3}{2,3}>{a}
\end{psmatrix}
}
\end{equation}\smallskip
\end{defi}

\begin{rem}
\label{ntpjp}\hfill
\begin{enumerate}[(1)]
\item
Homomorphisms of $H$-monoids are those monoid homomorphisms which are
also $H$-algebra homomorphisms.

\item
An $H$-monoid is called \textit{(completely) iterative} if its
underlying $H$-algebra has this property.
\end{enumerate}
\end{rem}

\begin{exa}
\label{ntptp}\hfill
\begin{enumerate}[(1)]
\item
$F_\lambda$~is an $H_\lambda$-monoid. Indeed, we know that
substitution yields the monoid structure (Remark~\ref{ds}) and tree
tupling yields the algebra structure (Example~\ref{dp}). Let us
consider the square
\vspace{3mm}
\[
\begin{psmatrix}[colsep=15mm,rowsep=10mm]
H_\lambda F_\lambda\otimes F_\lambda&
H_\lambda(F_\lambda\otimes F_\lambda)&HF_\lambda\\
F_\lambda\otimes F_\lambda&&F_\lambda
\psset{arrows=->,nodesep=3pt}
\everypsbox{\scriptstyle}
\ncline{1,1}{1,2}^{s}
\ncline{1,1}{2,1}<{\varphi\otimes\id}
\ncline{2,1}{2,3}_{m^F}
\ncline{1,2}{1,3}^{\rule{16pt}{0pt}Hm^F}
\ncline{1,3}{2,3}>{\varphi}
\end{psmatrix}
\]
\vspace{2mm}

\noindent
The elements~$t$ of $H_\lambda F\otimes F_\lambda$ in
context~$\Gamma$ are those of
\[H_\lambda F_\lambda(\Gamma_0)=F_\lambda(\Gamma_0)\times
F_\lambda(\Gamma_0)+F_\lambda\bigl(\Gamma_0+\{x\}\bigr)\]
for a given context~$\Gamma_0$
together with a substitution $f\colon\Gamma_0\to F_\lambda(\Gamma)$.
In case of the summand $F_\lambda(\Gamma_0)\times
F_\lambda(\Gamma_0)$ the lower passage
$m^F_\Gamma(\varphi_\Gamma\otimes\id)$ assigns to $t=(t_1,t_2)$ the
term $t_1\zav t_2$ with variables substituted according to~$f$. And
the upper passage first substitutes to $t_1$ and~$t_2$ according
to~$f$ separately, and then forms~$\zav$; the result is the same. In
case of the summand $F_\lambda(\Gamma_0+\{x\})$ the lower passage
assigns to~$t$ the term~$\lambda x\tecd t$ with variables substituted
according to~$f$; the upper one first substitutes in~$t$ and then
forms~$\lambda x\tecd{-}$ yielding the same result again.

\item
More generally, for every signature~\Sig\ we have an
$H_{\lambda,\Sig}$-monoid $F_{\lambda,\Sig}$.
\end{enumerate}
\end{exa}

\begin{thm}[see~\cite{FPT}]
\label{hnts}
The presheaf~$F_{\lambda,\Sig}$ of finite $\lambda$-\Sig-terms is the
initial $H_{\lambda,\Sig}$-monoid.
\end{thm}

\begin{thm}[see \cite{MU}]
\label{ntsjt}
The presheaf $T_{\lambda,\Sigma}$ of $\lambda$-\Sig-terms is an
$H_{\lambda,\Sigma}$-monoid with simultaneous substitution as monoid
structure.
\end{thm}

Although in~\cite{MU}, Example~13, just $T_\lambda$ is used, the
methods of that paper apply to $T_{\lambda,\Sig}$ immediately. The
following theorem proves a stronger property of~$T_{\lambda,\Sig}$,
corresponding to Theorem~\ref{hnts} above.

\begin{thm}
\label{hntss}
The presheaf~$T_{\lambda,\Sig}$ of $\lambda$-\Sig-terms is the
initial completely iterative $H_{\lambda,\Sig}$-monoid.
\end{thm}

An elementary proof of this theorem was presented in~\cite{AMV2}.
Here we will prove a more general result in Theorem~\ref{nncjjjp}
below.

\section{The Initial Iterative \texorpdfstring{$H$}{H}-Monoid}
\setcounter{equation}{0}
\label{snc}

\noindent The aim of this section is to prove that the
presheaf~$R_{\lambda,\Sig}$ of rational $\lambda$-\Sig-terms is the
initial iterative $H_{\lambda,\Sig}$-monoid in~\SetF. We have (in
contrast to the characterization of~$T_\lambda$ in the preceding
section) no elementary proof. Rather, we need to work with the
monad~$\rr_{\lambda,\Sig}$
of free iterative $H_{\lambda,\Sig}$-algebras on~\SetF\
(for which $R_{\lambda,\Sig}$~is~$\rr_{\lambda,\Sig}(V)$) and prove
that it is point-strong and use this strength further. We will
actually work in a more general setting (which can be
applied later for the case of typed $\lambda$-calculus).

\begin{asm}
\label{nncj}\textnormal{%
Throughout this section we assume that $H$~is a finitary endofunctor
of~\WW\ where
\begin{enumerate}[(1)]
\item
\WW~is a locally finitely presentable category, i.e., a cocomplete
category with a set of finitely presentable
objects~$\WW_{\mathsf{fp}}$ whose
closure under filtered colimits is all of~\WW.
\item
\WW~is also a strict
monoidal category with the unit~$I$ finitely presentable and the
tensor product preserving finite presentability: if $A,B$~are
finitely presentable, then so is $A\otimes B$.
\item
\WW~is right distributive, that is, for every object~$W$ the
endofunctor~${-\otimes}W$ preserves finite coproducts.
\item
The tensor product is a finitary functor, i.e., its preserves
filtered colimits (in both variables).
\end{enumerate}
We call categories satisfying \mbox{\textnormal{(1)--(4)}}
\textit{monoidally locally finitely presentable}.}
\end{asm}

\begin{exa}
\label{ncdjp}
Set is, as a cartesian closed category, monoidally locally finitely
presentable. For every monoidally locally finitely presentable
category~\WW\ all functor categories~$\WW^{\AAA}$, \AAA~small, have
the property too; for example, \SetF~with the cartesian product as
tensor.
However, in our paper we only use
the fact that \SetF~is a monoidally locally finitely presentable
w.r.t.~$\otimes$ in Remark~\ref{hndj}(i).
This follows from the fact that this is
equivalent to~$\Fin(\Set,\Set)$ with the tensor product given by
composition.
Observe that $\otimes$~is right distributive (since precomposition
with a given functor preserves colimits) but not left distributive.
\end{exa}

\begin{notation}
\label{nncd}
For every object~$Y$ of~\WW\ we denote by
\[
\varrho_Y\colon HRY\to RY\qquad\text{and}\qquad \eta_Y\colon Y\to RY
\]
the structure and oniversal morphism of
the free iterative $H$-algebra on~$Y$, which exists as proved
in~\cite{AMV1}.
This gives rise to the monad
\[\rr=(R,\eta,\mu)\]
where $\mu_X\colon RRY\to RY$ is the unique homomorphism
extending~$\eta_Y$:\smallskip

\begin{equation}
\label{nnrcj}
\myvc{
\begin{psmatrix}[colsep=15mm,rowsep=10mm]
HRRY&RRY&Y\\
HRY&RY
\psset{arrows=->,nodesep=3pt}
\everypsbox{\scriptstyle}
\ncline{1,1}{1,2}^{\varrho_{RY}}
\ncline{1,3}{1,2}^{\eta_{RY}}
\ncline{1,1}{2,1}<{H\mu_Y}
\ncline{1,3}{2,2}>{\eta_Y}
\ncline{2,1}{2,2}_{\varrho_Y}
\ncline{1,2}{2,2}<{\mu_Y}
\end{psmatrix}
}
\end{equation}\smallskip

\noindent
\rr~is called the \textit{rational monad} of the endofunctor~$H$.
\end{notation}

\begin{rem}
\label{nnct}
In~\cite{AMV1}
we described the free iterative $H$-algebra~$RY$ as the colimit of
the diagram of all ``flat equation''  morphisms
\[e\colon W\to HW+Y,\qquad\text{$W\in\WW$ finitely presentable,}\]
whose connecting morphisms (``equation morphisms'') are just the
coalgebra homomorphisms~$h$ for the endofunctor $H({-})+Y$:\smallskip

\begin{equation}
\label{nnrcd}
\myvc{
\begin{psmatrix}[colsep=15mm,rowsep=10mm]
W&HW+Y\\
W'&HW'+Y
\psset{arrows=->,nodesep=3pt}
\everypsbox{\scriptstyle}
\ncline{1,1}{1,2}^{e\rule{8pt}{0pt}}
\ncline{1,1}{2,1}<{h}
\ncline{2,1}{2,2}_{e'\rule{8pt}{0pt}}
\ncline{1,2}{2,2}>{Hh+\id}
\end{psmatrix}
}
\end{equation}\smallskip

\noindent
More detailed:
\begin{enumerate}[(i)]
\item
The category~$\mathsf{EQ}_Y$
of all flat equation morphisms in~$Y$ is filtered. The filtered
diagram
\[\mathsf{Eq}_Y\colon \mathsf{EQ}_Y\to W,\qquad
\mathsf{Eq}_Y(W\nsi{e}HW+Y)=W\]
has a colimit~$RY$ with the colimit injections
$e^\#\colon W\to RY$.

\item
For the flat equation morphism
$\inl\colon Y\to HY+Y$
put
\begin{equation}
\label{nnrct}
\eta_Y=\inr^\#\colon Y\to RY.
\end{equation}

\item
There is a unique isomorphism $i\colon RY\to HRY+Y$ such that the
squares\smallskip
\begin{equation}
\label{nnrcc}
\myvc{
\begin{psmatrix}[colsep=15mm,rowsep=10mm]
W&HW+Y\\
RY&HRY+Y
\psset{arrows=->,nodesep=3pt}
\everypsbox{\scriptstyle}
\ncline{1,1}{1,2}^{e\rule{8pt}{0pt}}
\ncline{1,1}{2,1}<{e^\#}
\ncline{2,1}{2,2}_{i_Y\rule{8pt}{0pt}}
\ncline{1,2}{2,2}>{He^\#+Y}
\end{psmatrix}
}
\end{equation}\smallskip

\noindent
commute for all flat equations~$e$. Put
\[\varrho\equiv HRY\nsi{\inl}HRY+Y\nsi{i_Y^{-1}}RY.\]
Then $RY$ together with $\eta_Y$ and~$\varrho_Y$ is the free
iterative $H$-algebra on~$Y$. We also have
\begin{equation}
\label{nnrcp}
i_Y=[\varrho_Y,\eta_Y]^{-1}.
\end{equation}
Furthermore, $e^\#$~is the unique coalgebra homomorphism from~$e$
to~$i_Y$.

\item
For every $e\colon W\to HW+Y$ the morphism $e^\#\colon W\to RY$ is
the unique solution (in the iterative algebra~$RY$) of
\[\eta_Y\bullet e\equiv W\nsi{e}HW+Y\nsi{HW+\eta_Y}HW+RY.\]

\item[(v)]
Let
\[e\colon W\to HW+RY\qquad\text{and}\qquad e'\colon W'\to HW'+RY\]
be two equation morphisms with $W$ and~$W'$ finitely presentable, and
let
$h$~be a coalgebra homomorphism from~$(W,e)$ to~$(W,e')$. Then
for the unique solutions of $e$ and~$e'$
we have
\[e^\dag=\bigl(e'\bigr)^\dag\tec h\colon W\to RY.\]

\item
Suppose we have two morphisms
\[f\colon V\to HV+W\qquad\text{and}\qquad e\colon W\to HW+RY\]
where $V,W$~are finitely presentable. Then we can form an equation
morphism
\vspace{3mm}
\[\begin{psmatrix}[colsep=15mm,rowsep=10mm]
e\ctv f\equiv V+W&
HV+W&HV+HW+RY\\
&&H(V+W)+RY
\psset{arrows=->,nodesep=3pt}
\everypsbox{\scriptstyle}
\ncline{1,1}{1,2}^{\rule{4mm}{0mm}[f,\inr]}
\ncline{1,2}{1,3}^{HV+e\rule{5mm}{0mm}}
\ncline{1,3}{2,3}>{\can+RY}
\end{psmatrix}\]
and we have
\begin{equation}
\label{rcoo}
\bigl(e^\dag\bullet f\bigr)^\dag=(e\ctv f)^\dag\tec\inl,
\end{equation}
see~\cite{AMV-Elgot}.

\item
Finally, every homomorphism~$h\colon A\to B$
of $H$-algebras between
iterative algebras $A$ and~$B$ \textit{preserves solutions}:
\[h\tec e^\dag=(h\bullet e)^\dag\colon X\to B\]
for every equation morphism $e\colon X\to HX+A$.
\end{enumerate}
\end{rem}

\begin{exa}
\label{nncc}
The rational monad of~$H_\lambda$ is the monad~$\rr_\lambda$ of
rational $\lambda$-terms with
constants: to every presheaf~$Y$ it assigns the
presheaf~$R_\lambda(Y)$ defined precisely as~$R_\lambda$ in
Example~\ref{dd}(v) except that in every
context~$\Gamma$ we can also use
elements of~$Y(\Gamma)$ to label the leaves.

More detailed: we first define the set~$R'_\lambda(Y)(\Gamma)$ of
rational trees in context~$\Gamma$
with constants from~$Y$. It consists of all rational trees
of the form~\eqref{rdj} such that
\[
\text{a node labelled by an element of~$Y(\Gamma)$ is a leaf.}
\]
By using the $\alpha$-conversion precisely as in
Example~\ref{dd}(iii), we
obtain the desired presheaf
\[R_\lambda(Y)(\Gamma)=
R'_\lambda(Y)(\Gamma)/{\sim_\alpha}.\]
It is again pointed; the pointing
$i^{R_\lambda(Y)}\colon V\to R_\lambda(Y)$
assigns to every variable the corresponding singleton tree. And
$R_\lambda(Y)$~is canonically an $H_\lambda$-algebra. We define
\[\eta_Y\colon Y\to R_\lambda(Y)\]
to assign to every element of~$Y(\Gamma)$ the corresponding singleton
tree. This is the free iterative $H_\lambda$-algebra on~$Y$, the
proof is completely analogous to that of Theorem~\ref{hdjt}.
\end{exa}

\begin{defi}
\label{nncp}
A \textit{point-strong monad} is a monad $\MM=(M,\eta,\mu)$ on~\WW\
together with a point-strength
\[s_{(X,x),(Y,y)}\colon(MX)\otimes Y\to M(X\otimes Y)\]
see Definition~\ref{hntj}, such that $s$~preserves the unit:\smallskip

\begin{equation}
\label{nnnrcp}
\myvc{
\begin{psmatrix}[colsep=15mm,rowsep=10mm]
MX\otimes Y&&M(X\otimes Y)\\
&X\otimes Y
\psset{arrows=->,nodesep=3pt}
\everypsbox{\scriptstyle}
\ncline{1,1}{1,3}^{s_{(X,x),(Y,y)}}
\ncline{2,2}{1,1}<{\eta^{\strut}_X\otimes Y}
\ncline{2,2}{1,3}>{\eta^{\strut}_{X\otimes Y}}
\end{psmatrix}
}
\end{equation}\smallskip

\noindent and the multiplication:\smallskip

\begin{equation}
\label{nnrcs}
\myvc{
\begin{psmatrix}[colsep=15mm,rowsep=10mm]
MMX\otimes Y&M(MX\otimes Y)&MM(X\otimes Y)\\
MX\otimes Y&&M(X\otimes Y)
\psset{arrows=->,nodesep=3pt}
\everypsbox{\scriptstyle}
\ncline{1,1}{1,2}^{s_{(MX,\eta_X\tec x),(Y,y)}}
\ncline{1,2}{1,3}^{Ms_{(X,x),(Y,y)}}
\ncline{1,1}{2,1}<{\mu_X\otimes Y}
\ncline{2,1}{2,3}_{s_{(X,x),(Y,y)}}
\ncline{1,3}{2,3}>{\mu_Y}
\end{psmatrix}
}
\end{equation}\smallskip
\end{defi}

\begin{exa}
\label{E:monad-T}
By our assumption that $H$~be finitary we know that all terminal
coalgebras for $H({-})+X$ exist, this follows from~\cite{n}, see
also~\cite{nc}. Equivalently, all free completely iterative algebras
for~$H$ exist (cf. Example~\ref{ndoa}), and they
yield the object map of a monad
$\TT=(T,\eta^T,\mu^T)$. This monad is the free completely iterative
monad on the endofunctor~$H$, see~\cite{nn}. The monad multiplication
$\mu^T_X\colon TTX\to TX$ is the unique algebra homomorphism
extending~$\id_{TX}$, i.e., such that
\begin{equation}
\label{homo}
\mu^T_X\tec \tau^{\phantom{T}}_{TX}=\tau^{\phantom{T}}_X\tec
H\mu^T_X\qquad\text{and}\qquad
\mu^T_X\tec\eta^T_{TX}=\id_{TX}.
\end{equation}
\end{exa}

\begin{thm}
\label{T:monad-T}
The free completely iterative monad~\TT\ of a point-strong
endofunctor~$H$ is point-strong.
\end{thm}

\begin{trivlist}
\item[\hspace{\labelsep}\textbf{Remark.}]
The strength of~\TT\ will be proved to be the unique natural
transformation~$s^T$ for which the diagram\smallskip

\begin{equation}
\label{ffflat}
\myvc{
\begin{psmatrix}[colsep=15mm,rowsep=10mm]
HTX\otimes Y&H(TX\otimes Y)&HT(X\otimes Y)\\
TX\otimes Y&&T(X\otimes Y)\\
X\otimes Y
\psset{arrows=->,nodesep=3pt}
\everypsbox{\scriptstyle}
\ncline{1,1}{1,2}^{s^H}
\ncline{1,2}{1,3}^{Hs^T}
\ncline{1,1}{2,1}<{\tau_X\otimes Y}
\ncline{2,1}{2,3}^{s^T}
\ncline{1,3}{2,3}>{\tau^{\phantom{T}}_{X\otimes Y}}
\ncline{3,1}{2,1}<{\eta_X^T\otimes Y}
\ncline{3,1}{2,3}>{\eta_{X\otimes Y}^T}
\end{psmatrix}
}
\end{equation}\smallskip

\noindent commutes. Note that we have dropped the subscripts indicating the
components of the natural transformations $s^H$ and~$s^T$ above; from
now on we shall frequently do this when components of natural
transformations are clear from the context.
\end{trivlist}

\proof
(a)\ 
Let $(X,x)$ and~$(Y,y)$ be pointed objects.
For every morphism $f\colon X\otimes Y\to TZ$ there exists a unique
morphism $f^\flat\colon TX\otimes Y\to TZ$ such that that the diagram
\smallskip

\begin{equation}
\label{fflat}
\myvc{
\begin{psmatrix}[colsep=15mm,rowsep=10mm]
HTX\otimes Y&H(TX\otimes Y)&HTZ\\
TX\otimes Y&&TZ\\
X\otimes Y
\psset{arrows=->,nodesep=3pt}
\everypsbox{\scriptstyle}
\ncline{1,1}{1,2}^{s^H}
\ncline{1,2}{1,3}^{Hf^\flat}
\ncline{1,1}{2,1}<{\tau\otimes Y}
\ncline{2,1}{2,3}^{f^\flat}
\ncline{1,3}{2,3}>{\tau}
\ncline{3,1}{2,1}<{\eta\otimes Y}
\ncline{3,1}{2,3}>{f}
\end{psmatrix}
}
\end{equation}\smallskip

\noindent
commutes. Indeed, the algebra~$TZ$ is completely iterative. Due to
$(HTX+Y)\otimes Y=HTX\otimes Y+X\otimes Y$, see
Assumption~\ref{nncj}(3),
we obtain an equation
morphism in~$TZ$ as follows:
\[TX\otimes Y\nsi{[\tau_X^{\strut},\eta_X^T]^{-1}\otimes Y} HTX\otimes
Y+X\otimes Y\nsi{s^H_{X,Y}+f}H(TX\otimes Y)+TZ.\]
Its unique solution is denoted by~$f^\flat$. It is characterized by
the commutative diagram\smallskip

\[\begin{psmatrix}[colsep=30mm,rowsep=10mm]
TX\otimes Y&TZ\\
HTX\otimes Y+X\otimes Y\\
H(TX\otimes Y)+TZ&HTZ+TZ
\psset{arrows=->,nodesep=3pt}
\everypsbox{\scriptstyle}
\ncline{2,1}{1,1}<{[\tau,\eta^T]\otimes Y}
\ncline{2,1}{3,1}<{s^H+f}
\ncline{1,1}{1,2}^{f^\flat}
\ncline{3,1}{3,2}_{Hf^\flat+TZ}
\ncline{3,2}{1,2}>{[\tau,TZ]}
\end{psmatrix}\]\smallskip

\noindent
It is easy to verify that this diagram commutes iff
\eqref{fflat}~does.\smallskip

\noindent(b)\ 
Put
\[s^T_{(X,x),(Y,y)}=\bigl(\eta^T_{X\otimes Y}\bigr)^\flat\colon
TX\otimes Y\to
T(X\otimes Y).\]
In other words, we define the components of~$s^T$ via
\eqref{fflat} uniquely.\smallskip

(b1)\
$s^T$~is natural: the squares\smallskip

\[\begin{psmatrix}[colsep=30mm,rowsep=10mm]
TX\otimes Y&T(X\otimes Y)\\
T'X\otimes Y'&T(X'\otimes Y')
\psset{arrows=->,nodesep=3pt}
\everypsbox{\scriptstyle}
\ncline{1,1}{1,2}^{s^T}
\ncline{1,1}{2,1}<{Tg\otimes h}
\ncline{1,2}{2,2}>{T(g\otimes h)}
\ncline{2,1}{2,2}_{s^T}
\end{psmatrix}\]\smallskip

\noindent
commute for all morphisms $g$ and~$h$ of~$I/\WW$ since both
passages
form~$f^\flat$ for
\[f=\eta^{T}_{X'\otimes Y'}\tec(g\otimes h)\colon X\otimes Y\to
T(X'\otimes Y').\]
Indeed, for the upper passage, $f^\flat=T(g\otimes h)\tec s^T$, use
the following diagram:\smallskip

\[\begin{psmatrix}[colsep=18mm,rowsep=15mm]
HTX\otimes Y&HTX\otimes Y&HT(X\otimes Y)&
HT(X'\otimes Y')\\
TX\otimes Y&&T(X\otimes Y)&T(X'\otimes Y')\\
X\otimes Y&&&X'\otimes Y'
\psset{arrows=->,nodesep=3pt}
\everypsbox{\scriptstyle}
\ncline{1,1}{1,2}^{s^H}
\ncline{1,2}{1,3}^{Hs^T}
\ncline{1,3}{1,4}^{HT(h\otimes g)}
\ncline{1,1}{2,1}<{H\tau^{\strut}_X\otimes Y}
\ncline{1,3}{2,3}>{\tau^{\strut}_{X\otimes Y}}
\ncline{1,4}{2,4}>{\tau^{\strut}_{X'\otimes Y'}}
\ncline{2,1}{2,3}^{s^T}
\ncline{2,3}{2,4}^{T(h\otimes g)}
\ncline{3,1}{2,1}<{\eta^T_X\otimes Y}
\ncline{3,1}{2,3}>{\eta^T_{X\otimes Y}}
\ncline{3,4}{2,4}>{\eta^{T}_{X'\otimes Y'}}
\ncline{3,1}{3,4}_{h\otimes g}
\end{psmatrix}\]\smallskip

\noindent
The two left-hand parts form
Diagram~\eqref{fflat}, the remaining two
commute by naturality of $\tau$ and~$\eta$.

The lower passage $f^\flat=s^T\tec(Tg\otimes h)$ follows from the
following diagram:\smallskip

\[\begin{psmatrix}[colsep=18mm,rowsep=15mm]
HTX\otimes Y&HTX'\otimes Y'&H(T(X')\otimes Y')&
HT(X'\otimes Y')\\
TX\otimes Y&TX'\otimes Y'&&T(X'\otimes Y')\\
X\otimes Y&X'\otimes Y'
\psset{arrows=->,nodesep=3pt}
\everypsbox{\scriptstyle}
\ncline{1,1}{1,2}^{HTg\otimes h}
\ncline{1,2}{1,3}^{s^H}
\ncline{1,3}{1,4}^{Hs^T}
\ncline{1,1}{2,1}<{\tau^{\strut}_X\otimes Y}
\ncline{1,2}{2,2}>{\tau^{\strut}_{X'}\otimes Y'}
\ncline{1,4}{2,4}>{\tau^{\strut}_{X'\otimes Y'}}
\ncline{2,1}{2,2}^{Tg\otimes h}
\ncline{2,2}{2,4}^{s^T}
\ncline{3,1}{2,1}<{\eta^T_X\otimes Y}
\ncline{3,2}{2,2}>{\eta^T_{X'}\otimes Y'}
\ncline{3,2}{2,4}>{\eta^T_{X'\otimes Y'}}
\ncline{3,1}{3,2}_{g\otimes h}
\end{psmatrix}\]\smallskip

\noindent
The right-hand parts form
Diagram~\eqref{strength-T}, the left-hand ones
commute by naturality of $\tau$ and~$\eta^T$.\smallskip

(b2)\
$s^T$~is a point-strength of the endofunctor~$T$. Indeed, the axiom
\begin{equation}
\label{strength-T}
s^T_{(X,x)(V,v)}=\id_{T(X)}
\end{equation}
follows from the fact that if $(Y,y)=(V,\id)$,
then Diagram~\eqref{ffflat} commutes
with~$\id_{T(X)}$ in lieu of~$s^T$. To verify the
Axiom~\eqref{triangle}, apply~(a) to $f=\eta^T_{X\otimes Y\otimes
Z}$: we prove that the lower passage of~\eqref{triangle} serves
as~$f^\flat$. In detail, the diagram\smallskip

\[\begin{psmatrix}[colsep=24mm,rowsep=12mm]
&HT(X\otimes Y)\otimes Z&H(T(X\otimes Y)\otimes Z)\\
HTX\otimes Y\otimes Z&H(TX\otimes Y)\otimes Z&
HT(X\otimes Y\otimes Z)\\
TX\otimes Y\otimes Z&T(X\otimes Y)\otimes Z&
T(X\otimes Y\otimes Z)\\
X\otimes Y\otimes Z
\psset{arrows=->,nodesep=3pt}
\everypsbox{\scriptstyle}
\ncline{2,2}{1,2}<{Hs^T\otimes Z}
\ncline{1,2}{1,3}^{s^H}
\ncline{1,3}{2,3}>{Hs^T}
\ncline{2,1}{2,2}^{s^H\otimes Z}
\ncline{2,1}{3,1}<{\tau^{\strut}_X\otimes Y\otimes Z}
\ncarc[arcangle=75]{1,2}{3,2}>{\tau^{\strut}_{X\otimes Y}\otimes Z}
\ncline{2,3}{3,3}>{\tau^{\strut}_{X\otimes Y\otimes Z}}
\ncline{3,1}{3,2}^{s^T\otimes Z}
\ncline{3,2}{3,3}^{s^T}
\ncline{4,1}{3,1}<{\eta^T_X\otimes Y\otimes Z}
\ncline{4,1}{3,2}<{\eta^T_{X\otimes Y}\otimes Z}
\ncline{4,1}{3,3}>{\rule{0pt}{10pt}\eta^T_{X\otimes Y\otimes Z}}
\end{psmatrix}\]\smallskip

\noindent
commutes. Indeed, all inner parts commute by two applications
of~\eqref{ffflat}.\smallskip

(b3)\
It remains to verify the axioms of Definition~\ref{nncp}.
For~\eqref{nnnrcp} use the lower triangle of
Diagram~\eqref{ffflat}.
For~\eqref{nnrcs} apply~(a) to
\[f=s^T\colon TX\otimes Y\to T(X\otimes Y).\]
We prove that both passages of~\eqref{nnrcs} serve as~$f^\flat$. For
the lower passage, $(s^T)^\flat=s^T\tec\mu^T\otimes Y$, use the
following diagram\smallskip

\[\begin{psmatrix}[colsep=22mm,rowsep=12mm]
HTTX\otimes Y&HT(TX\otimes Y)&HT&HT(X\otimes Y)\\
&HTX\otimes Y\\
TTX\otimes Y&TX\otimes Y&&T(X\otimes Y)\\
TX\otimes Y
\psset{arrows=->,nodesep=3pt}
\everypsbox{\scriptstyle}
\ncline{1,1}{1,2}^{s^H}
\ncline{1,2}{1,3}^{\rule{6mm}{0mm}H(\mu^T\otimes Y)}
\ncline{1,3}{1,4}\naput{Hs^T}
\ncline{1,1}{2,2}<{H\mu^T\otimes Y}
\ncline{2,2}{1,3}>{s^H}
\ncline{2,2}{3,2}>{\tau\otimes Y}
\ncline{1,1}{3,1}<{\tau\otimes Y}
\ncline{1,4}{3,4}>{\tau}
\ncline{3,1}{3,2}^{\mu^T\otimes Y}
\ncline{3,2}{3,4}^{s^T}
\ncline{4,1}{3,1}<{\eta^TT\otimes Y}
\ncline[doubleline=true]{-}{4,1}{3,2}
\ncline{4,1}{3,4}>{s^T}
\end{psmatrix}\]\smallskip

\noindent
The upper left-hand part is
Equation~\eqref{homo}, the lower one commutes by the
monad axiom $\mu^T\tec\eta^TT=\id$, the upper triangle is the
naturality of~$s^H$, and the right-hand part follows
from~\eqref{ffflat}.

For the upper passage, $(s^T)^\flat=\mu^T\tec Ts^T\tec s^T$, use the
following diagram\smallskip

\[\begin{psmatrix}[colsep=18mm,rowsep=12mm]
H(TTX\otimes Y)\\
HTTX\otimes Y&HT(TX\otimes Y)&HTT(X\otimes Y)&HT(X\otimes Y)\\
TTX\otimes Y&T(TX\otimes Y)&TT(X\otimes Y)&T(X\otimes Y)\\
TX\otimes Y
\psset{arrows=->,nodesep=3pt}
\everypsbox{\scriptstyle}
\ncline{2,1}{1,1}<{s^H}
\ncline{1,1}{2,2}>{Hs^T}
\ncline{2,2}{2,3}^{HTs^T}
\ncline{2,3}{2,4}^{H\mu^T}
\ncline{2,1}{3,1}<{\tau\otimes Y}
\ncline{2,2}{3,2}>{\tau}
\ncline{2,3}{3,3}>{\tau}
\ncline{2,4}{3,4}>{\tau}
\ncline{3,1}{3,2}^{s^T}
\ncline{3,2}{3,3}^{Ts^T}
\ncline{3,3}{3,4}^{\mu^T}
\ncline{4,1}{3,1}<{\eta^T\otimes Y}
\ncline{4,1}{3,2}<{\eta^T}
\ncline{4,1}{3,4}>{\rule{8pt}{0pt}s^T}
\end{psmatrix}\]\smallskip

\noindent The three upper squares commute due to~\eqref{ffflat}, the
naturality of~$\tau$ and~\eqref{homo}. The lower triangles commute
due to~\eqref{ffflat}, the naturality of~$s^T$ and
$\mu^T\tec\eta^TT=\id$.\qed

\begin{rem}
\label{R:monad-T}
Recall from Example~\ref{ndoa} that $T=HT+\Id$ with injections
$\tau$ and~$\eta^T$. From the Diagram~\eqref{ffflat} we see that the
strength~$s^T$ then has the form
\[s^T=Hs^T\tec s^H+X\otimes Y\colon HTX\otimes Y+X\otimes Y\to
HT(X\otimes Y)+X\otimes Y.\]
\end{rem}

\begin{thm}
\label{nncs}
The rational monad of a point-strong endofunctor is point-strong.
\end{thm}

\begin{trivlist}
\item[\hspace{\labelsep}\textbf{Remark.}]
The strength of~\rr\ will be proved to be the unique natural
transformation~$s^R$ for which the diagram\smallskip

\begin{equation}
\label{nnrcsjp}
\myvc{
\begin{psmatrix}[colsep=15mm,rowsep=10mm]
HRX\otimes Y&H(RX\otimes Y)&HR(X\otimes Y)\\
RX\otimes Y&&R(X\otimes Y)\\
X\otimes Y
\psset{arrows=->,nodesep=3pt}
\everypsbox{\scriptstyle}
\ncline{1,1}{1,2}^{s^H}
\ncline{1,2}{1,3}^{Hs^R}
\ncline{1,1}{2,1}<{\varrho_X\otimes Y}
\ncline{2,1}{2,3}^{s^R}
\ncline{1,3}{2,3}>{\varrho_{X\otimes Y}}
\ncline{3,1}{2,1}<{\eta_X\otimes Y}
\ncline{3,1}{2,3}>{\eta_{X\otimes Y}}
\end{psmatrix}
}
\end{equation}\smallskip

\noindent
commutes.
\end{trivlist}

\proof
(a)\ 
Given pointed objects $(X,x)$ and~$(Y,y)$,
we prove that for every morphism $f\colon X\otimes Y\to RZ$ there
exists a unique morphism $f^\flat\colon RX\otimes Y\to RZ$ such that
the
following diagram commutes:\smallskip

\begin{equation}
\label{nnrcss}
\myvc{
\begin{psmatrix}[colsep=15mm,rowsep=10mm]
HRX\otimes Y&H(RX\otimes Y)&HRZ\\
RX\otimes Y&&RZ\\
X\otimes Y
\psset{arrows=->,nodesep=3pt}
\everypsbox{\scriptstyle}
\ncline{1,1}{1,2}^{s^H}
\ncline{1,2}{1,3}\naput{Hf^\flat}
\ncline{1,1}{2,1}<{\varrho_X\otimes Y}
\ncline{2,1}{2,3}^{f^\flat}
\ncline{1,3}{2,3}>{\varrho_Z}
\ncline{3,1}{2,1}<{\eta_X\otimes Y}
\ncline{3,1}{2,3}>{f}
\end{psmatrix}
}
\end{equation}\smallskip

(a1)\
Assume that $Y$~is finitely presentable. Recall
$RY=\colim\mathsf{Eq}_Y$ from Remark~\ref{nnct}.
For every object
\[e\colon W\to HW+X\qquad\text{in $\mathsf{EQ}_X$}\]
define, using the distributivity $(HW+Y)\otimes Y=HW\otimes
Y+X\otimes Y$ (see
Assumption~\ref{nncj}(3)), the equation morphism
\begin{equation}
\label{rcxxx}
\hat{e}\equiv W\otimes Y\nsi{e\otimes Y} HW\otimes Y+X\otimes Y
\nsi{s^H+f} H(W\otimes Y)+RZ.
\end{equation}
Since $W\otimes Y$ is finitely presentable by
Assumption~\ref{nncj}(2), we
obtain the unique solution $\hat{e}^\dag\colon W\otimes Y\to RZ$,
and those solutions
form a cocone of the diagram $\mathsf{Eq}_X\otimes Y$.
Indeed, given a
connecting morphism\smallskip

\[
\begin{psmatrix}[colsep=15mm,rowsep=10mm]
W&HW+X\\
W'&HW'+X
\psset{arrows=->,nodesep=3pt}
\everypsbox{\scriptstyle}
\ncline{1,1}{1,2}^{e\rule{5mm}{0mm}}
\ncline{1,1}{2,1}<{h}
\ncline{2,1}{2,2}_{e'\rule{5mm}{0mm}}
\ncline{1,2}{2,2}>{Hh+X}
\end{psmatrix}
\]\smallskip

\noindent
then $h\otimes Y$ is a coalgebra homomorphism from~$\hat{e}$
to~$\hat{e}'$:\smallskip

\[
\begin{psmatrix}[colsep=15mm,rowsep=10mm]
W\otimes Y&(HW\otimes Y)+(X\otimes Y)&
H(W\otimes Y)+RZ\\
W'\otimes Y&(HW'\otimes Y)+(X\otimes Y)&
H(W'\otimes Y)+RZ
\psset{arrows=->,nodesep=3pt}
\everypsbox{\scriptstyle}
\ncline{1,1}{1,2}^{e\otimes Y\rule{12mm}{0mm}}
\ncline{1,2}{1,3}^{s^H+f}
\ncline{1,1}{2,1}<{h\otimes Y}
\ncline{1,2}{2,2}>{(Hh\otimes Y)+(X\otimes Y)}
\ncline{2,1}{2,2}_{e'\otimes Y\rule{12mm}{0mm}}
\ncline{1,3}{2,3}>{H(h\otimes Y)+RZ}
\ncline{2,2}{2,3}_{s^H+f}
\end{psmatrix}
\]\smallskip

\noindent
which implies, by Remark~\ref{nnct}(v) that
\[\hat{e}^\dag=\widehat{e'}\rule{0pt}{11pt}^\dag\tec(h\otimes Y).\]
Consequently, we can define
\[f^\flat\colon RX\otimes Y\to RZ\]
by the commutativity of the triangles\smallskip

\begin{equation}
\label{nnrco}
\myvc{
\begin{psmatrix}[colsep=15mm,rowsep=10mm]
W\otimes Y\\
RX\otimes Y&RZ
\psset{arrows=->,nodesep=3pt}
\everypsbox{\scriptstyle}
\ncline{1,1}{2,1}<{e^\#\otimes Y}
\ncline{1,1}{2,2}>{\hat{e}^\dag}
\ncline{2,1}{2,2}_{f^\flat}
\end{psmatrix}
}\qquad\quad\text{for all $e\in \mathsf{EQ}_X$.}
\end{equation}\smallskip

\noindent
Indeed, since ${-\otimes}Y$ is a finitary functor by
Assumption~\ref{nncj}(4),
we see that $RX\otimes Y$ is a colimit
of~$\mathsf{Eq}_Y\otimes
Y$
with the colimit
cocone $e^\#\otimes Y$. We now verify that the Diagram~\eqref{nnrcss}
commutes. Consider the diagram below:\smallskip

\begin{equation}
\label{rcyyy}
\myvc{
\begin{psmatrix}[colsep=25mm]
W\otimes Y&RX\otimes Y&RZ\\
HW\otimes Y+X\otimes Y&HRX\otimes Y+X\otimes Y\\
H(W\otimes Y)+RZ&H(RX\otimes Y)+RZ&HRZ+RZ
\psset{arrows=->,nodesep=3pt}
\everypsbox{\scriptstyle}
\ncline{1,1}{2,1}<{e\otimes Y}
\ncline{1,1}{1,2}^{e^\#\otimes Y}
\ncline{1,2}{1,3}^{f^\flat}
\ncline{1,2}{2,2}<{i_X\otimes Y}
\ncline{2,1}{2,2}_{He^\#\otimes Y+X\otimes Y}
\ncline{2,1}{3,1}<{s^H+f}
\ncline{2,2}{3,2}>{s^H+f}
\ncline{3,1}{3,2}_{H(e^\#\otimes Y)+RZ}
\ncline{3,2}{3,3}_{Hf^\flat+RZ}
\ncline{3,3}{1,3}<{[\varrho_Z,RZ]}
\ncarc[arcangle=-30,linestyle=dashed]{2,2}{1,2}>{[\varrho_X\otimes
Y,\eta_X\otimes Y]}
\end{psmatrix}
}
\end{equation}\smallskip

\noindent
Notice first that the left-hand
edge is~$\hat{e}$. The upper left-hand
part commutes by~\eqref{nnrcc}, and the lower one does by
naturality of~$s^H$. The outside of the diagram commutes since
$f^\flat\tec(e^\#\otimes Y)$ is the unique solution of~$\hat{e}$ in
the iterative algebra~$RZ$.
Thus, the right-hand part commutes when precomposed by any
$e^\#\otimes Y$. So since the latter morphisms are collectively
epimorphic (being the injections of $\colim
\mathsf{Eq}_X\otimes Y$), we
see that the right-hand part commutes. Now we use that $i_X$~is an
isomorphism with the inverse~$[\varrho_X,\eta_X]$, see
Equation~\eqref{nnrcp},
which implies
\[[\varrho_X\otimes Y,\eta_X\otimes Y]=(i_X\otimes Y)^{-1}.\]
Finally observe that the two coproduct components of the right-hand
part of~\eqref{rcyyy} yield
precisely the upper and lower parts
of~\eqref{nnrcss}---this proves that \eqref{nnrcss}~commutes.

It only remains to prove the uniqueness of~$f^\flat$. So suppose we
have some $f^\flat$ such that Diagram~\eqref{nnrcss}
commutes. Equivalently,
the right-hand part of~\eqref{rcyyy} commutes, and this implies that
$f^\flat\tec(e^\#\otimes Y)$ is, for every~$e$ in~$\mathsf{EQ}_X$, a
solution
of~$\hat{e}$. This determines~$f^\flat$ uniquely.\smallskip

(a2)\ 
Let $Y$~be arbitrary. Then since \WW~is locally finitely presentable
we can express~$Y$ as a filtered colimit
\[Y=\colim_{q\in Q}Y^q\qquad\text{with colimit cocone $y^q\colon
Y^q\to Y$}\]
of finitely presentable objects~$Y^q$.
By Assumption~\ref{nncj}(2) the unit object~$I$
is finitely presentable, thus the given pointing of~$Y$:
\[y\colon I\to\colim_{q\in Q}Y^q\]
factorizes through some~$y^q$. The diagram above being filtered, we
can assume that this factorization takes place for every $q\in Q$, in
other words, that we have a filtered diagram of pointed objects~$Y^q$
with colimit~$Y$ (and with all the connecting morphisms $Y^q\to
Y^{q'}$ preserving the pointing).

Given $f\colon X\otimes Y\to RZ$, for every $q\in Q$ we know from the
previous part~(a1) that there exists a unique
\[f_q^\flat\colon RX\otimes Y^q\to RZ\]
such that Diagram~\eqref{nnrcss}
commutes when $f^\flat$~is replaced
by~$f^\flat_q$
and $f$~by
\[f_q\equiv X\otimes Y^q\nsi{X\otimes y^q}X\otimes Y\nsi{f}RZ.\]
This defines a unique $f^\flat\colon RX\otimes Y\to RZ$ with
\begin{equation}
\label{nrcjsa}
f^\flat_q=f^\flat\tec(RX\otimes y^q)\qquad\text{for all $q\in Q$.}
\end{equation}
Now
Diagram~\eqref{nnrcss} commutes because $HRX\otimes Y=\colim_{q\in Q}
HRX\otimes Y^q$ as well as $X\otimes Y=\colim_{q\in Q} X\otimes Y^q$.
And $f^\flat$~is uniquely determined by this commutativity; indeed,
for any~$f^\flat$ such that \eqref{nnrcss}~commutes one easily
verifies that~\eqref{nrcjsa} holds using the uniqueness
of~$f^\flat_q$ from part~(a1).\smallskip

\noindent(b)\
Analogously to the proof of Theorem~\ref{T:monad-T} put
\begin{equation}
\label{nnrcojp}
s^R_{(X,x),(Y,y)}=\eta^\flat_{X\otimes Y}\colon RX\otimes Y\to
R(X\otimes
Y).
\end{equation}
The verification that $s^R$~is the desired strength is analogous to
the above proof: just replace $T$ by~$R$ (and $\tau$ by~$\varrho$).\qed

\begin{rem}
\label{R:new}
The proofs of Theorems \ref{T:monad-T} and~\ref{nncs} have the same
structure, and also the proof that the
monad~$\mathbb{F}_{\lambda,\Sig}$ is point-strong can proceed
analogously:

Let $H$~be a point-strong endofunctor of~\WW\ and let
$(\hat{M},\hat{\mu},\hat{\eta})$~be a monad. Suppose that a natural
transformation $\alpha\colon H\hat{M}\to\hat{M}$ has the property
that for every morphism $f\colon X\otimes Y\to MZ$ there exists a
unique morphism $f^\flat\colon MX\otimes Y\to MZ$ with
$f=f^\flat\tec(\hat{\eta}^{\strut}_X\otimes Y)$ and
$f^\flat\tec(\alpha^{\strut}_X\otimes Y)=\alpha^{\strut}_X\tec
Hf^\flat\tec s^H$. Then $M$~is a point-strong monad w.r.t.
$s^M=\hat{\eta}^\flat_{X\otimes Y}$.
\end{rem}

\begin{rem}
\label{copr}
The morphisms
\[\varrho^{\strut}_X\colon HRX\to
RX\qquad\text{and}\qquad \eta^{\strut}_X\colon X\to
RX\]
of~\eqref{nnrcsjp} are coproduct injections of
\[RX=HRX+X\]
as proved in~\cite{AMV1}. From diagram~\eqref{nnrcsjp} we conclude
that the strength of~\rr,
\[s^R\colon RX\otimes Y\to R(X\otimes Y)\]
whose domain is $HRX\otimes Y+X\otimes Y$ by~\ref{nncj}(2) and
codomain is $HR(X\otimes Y)+X\otimes Y$, has the form
\[s^R=Hs^R\tec s^H+X\otimes Y.\]
\end{rem}

\begin{cor}
\label{nncss}
For a point-strong endofunctor~$H$
the free iterative $H$-algebra
\[RI\]
on the unit object is an $H$-monoid w.r.t. the unit
$i=\eta_I\colon I\to RI$ and
the multiplication
\begin{equation}
\label{multi}
m\equiv RI\otimes RI\nsi{s^R_{I,RI}} RRI\nsi{\mu_I} RI.
\end{equation}
\end{cor}

\begin{proof}
Indeed, the unit laws are obvious:\smallskip
\[
\begin{psmatrix}[colsep=15mm, rowsep=10mm]
&I\otimes RI&RI\\
&RI\otimes RI&RRI&RI\\
&RI\otimes I&RI
\psset{arrows=->,nodesep=3pt}
\everypsbox{\scriptstyle}
\ncline{1,2}{2,2}<{\eta_I\otimes RI}
\ncline{1,2}{2,2}>{\rule{8mm}{0mm}\eqref{nnrcsjp}}
\ncline{1,3}{2,3}<{\eta_{RI}}
\ncline[doubleline=true,arrows=-]{1,3}{2,4}
\ncline{2,2}{2,3}_{s^R_{I,RI}}
\ncline{2,3}{2,4}_{\mu_I}
\ncline{3,2}{2,2}<{RI\otimes\eta_I}
\ncline{3,3}{2,3}<{R\eta_I}
\ncline{3,2}{3,3}_{s^R_{I,I}}
\ncline[doubleline=true,arrows=-]{3,3}{2,4}
\ncline[doubleline=true,arrows=-]{1,2}{1,3}
\end{psmatrix}
\]\smallskip

\noindent
where the lower square commutes by the
naturality of~$s^R$. For the associativity
we have
the following commutative diagram
\vspace{6mm}
\[
\rule{7mm}{0mm}\begin{psmatrix}[colsep=18mm]
RI\otimes RI\otimes RI&&RI\otimes RRI&RI\otimes RI\\
RRI\otimes RI&R(RI\otimes RI)&RRRI&RRI\\
RI\otimes RI&&RRI&RI
\psset{arrows=->,nodesep=3pt}
\everypsbox{\scriptstyle}
\ncline{1,1}{1,3}_{RI\otimes s^R_{I,RI}}
\ncline{1,3}{1,4}_{RI\otimes\mu_I}
\ncline{1,1}{2,1}>{s^R_{I,RI}\otimes RI}
\ncline{1,3}{2,3}<{s^R_{I,RRI}}
\ncline{1,4}{2,4}<{s^R_{I,RI}}
\ncline{2,1}{2,2}^{s^R_{RI,RI}}
\ncline{2,2}{2,3}^{Rs^R_{I,RI}}
\ncline{2,3}{2,4}^{R\mu_I}
\ncline{2,1}{3,1}>{\mu_I\otimes RI\rule{24mm}{0mm}\eqref{nnrcs}}
\ncline{2,3}{3,3}<{\mu_{RI}}
\ncline{3,1}{3,3}^{s^R_{I,RI}}
\ncline{3,3}{3,4}^{\mu_I}
\ncline{2,4}{3,4}<{\mu_I}
\ncarc[arcangle=10]{1,1}{1,4}^{RI\otimes m}
\ncarc[arcangle=-10]{3,1}{3,4}_{m}
\ncarc[arcangle=-57]{1,1}{3,1}<{m\otimes RI}
\ncarc[arcangle=30]{1,4}{3,4}>{m}
\ncline{1,1}{2,2}>{s^R_{I,RI\otimes RI}}
\end{psmatrix}
\]
\vspace{5mm}

\noindent
Finally, the Diagram~\eqref{abb} commutes due to~\eqref{nnrcsjp}:
\vspace{3mm}
\[
\begin{psmatrix}[colsep=20mm]
HRI\otimes RI&H(RI\otimes RI)&HRRI&HRI\\
RI\otimes RI&&RRI&RI
\psset{arrows=->,nodesep=3pt}
\everypsbox{\scriptstyle}
\ncline{1,1}{1,2}^{s^H}
\ncline{1,2}{1,3}_{Hs^R}
\ncline{1,3}{1,4}_{H\mu}
\ncline{1,1}{2,1}<{\varrho\otimes RI}
\ncline{1,3}{2,3}>{\varrho R}
\ncline{1,4}{2,4}>{\varrho}
\ncline{2,1}{2,3}^{s^R}
\ncline{2,3}{2,4}^{\mu}
\ncarc[arcangle=10]{1,2}{1,4}^{Hm}
\ncarc[arcangle=-10]{2,1}{2,4}_{m}
\end{psmatrix}
\]
\end{proof}
\vspace{2mm}

\begin{cor}
\label{cojp}
The free completely iterative $H$-algebra
\[TI\]
on the unit object is
an $H$-monoid w.r.t. $\eta^T_I$ and~$\mu^T_I\tec s^T_{I,TI}$.
\end{cor}

The proof is completely analogous to the previous one.

\begin{notation}
\label{nnco}
Let
\[e\colon W\to HW+I\]
be a flat equation morphism in~$I$, and let a pointing of~$W$ be
given.
\begin{enumerate}[(i)]
\item
$e^\#\colon W\to RI$ denotes the colimit morphism of $RI=\colim
\mathsf{Eq}_I$ and
\[\widehat{e^\#}\colon RW\to RI\]
its unique extension to a homomorphism of $H$-algebras.
\item
$e*W$~denotes the following equation morphism in~$W$:
\[e*W\equiv W\otimes W\nsi{e\otimes W} HW\otimes W+W \nsi{s^H+W}
H(W\otimes W)+W.\]
\item
$\langle e\rangle\colon W\otimes W+W\to H(W\otimes W+W)+I$ denotes
the flat equation morphism whose left-hand component is
\vspace{3mm}
\[\begin{psmatrix}[colsep=15mm,rowsep=12mm]
\langle e\rangle\tec\inl\equiv W\otimes W&
H(W\otimes W)+W\\
&H(W\otimes W)+HW+I&H(W\otimes W+W)+I
\psset{arrows=->,nodesep=3pt}
\everypsbox{\scriptstyle}
\ncline{1,1}{1,2}^{e*W}
\ncline{1,2}{2,2}<{H(W\otimes W)+e}
\ncline{2,2}{2,3}^{\can+I}
\end{psmatrix}\]
and the right-hand one is
\[
\langle e\rangle\tec\inr\equiv W\nsi{e} HW+I\nsi{H\inr+I} H(W\otimes
W+W)+I.
\]
\end{enumerate}
\end{notation}

\begin{lem}
\label{nncdv}
For every flat equation morphism $e\colon W\to HW+I$ with $W$~pointed
the square\smallskip

\begin{equation}
\label{nnrcjn}
\myvc{
\begin{psmatrix}[colsep=15mm,rowsep=10mm]
W\otimes W&RW\\
W\otimes W+W&RI
\psset{arrows=->,nodesep=3pt}
\everypsbox{\scriptstyle}
\ncline{1,1}{1,2}^{\rule{3mm}{0mm}(e*W)^\#}
\ncline{1,2}{2,2}>{\widehat{e^\#}}
\ncline{1,1}{2,1}<{\inl}
\ncline{2,1}{2,2}_{\rule{7mm}{0mm}\langle e\rangle^\#}
\end{psmatrix}
}
\end{equation}\smallskip

\noindent commutes.
\end{lem}

\begin{proof}
Notice that $\langle e\rangle$~is precisely of the form $e\ctv f$
from Remark~\ref{nnct}(vi) for $f=e*W$. Also recall from
Remark~\ref{nnct}(iv) that  $\langle e\rangle^\#=(\eta\bullet\langle
e\rangle)^\dag$ and similarly $(e*W)^\#=(\eta\bullet(e*W))^\dag$.
We shall also use that the
$H$-algebra homomorphism~$\widehat{e^\#}$ preserves solutions (cf.
Remark~\ref{nnct}(vii)). Thus, we compute
\begin{align*}
\widehat{e^\#}\tec(e*W)^\#&=
\widehat{e^\#}\tec\bigl(\eta\bullet(e*W)\bigr)^\dag&
\text{\ref{nnct}(iv)}\\
&=\bigl((\widehat{e^\#}\tec\eta)\bullet(e*W)\bigr)^\dag&
\text{\ref{nnct}(vii)}\\
&=\bigl(e^\#\bullet(e*W)\bigr)^\dag&
\text{$\widehat{e^\#}$ extends $e^\#$}\\
&=\bigl((\eta\tec e)^\dag\bullet(e*W)\bigr)^\dag&
\text{\ref{nnct}(iv)}\\
&=(\eta\tec e\ctv e*W)^\dag\tec\inl&
\text{\eqref{rcoo}}\\
&=\Bigl(\eta\bullet\bigl(e\ctv(e*W)\bigr)\Bigr)^\dag\tec\inl&
\text{obvious}\\
&=\bigl(e\ctv(e*W)\bigr)^\#\tec\inl&
\text{\ref{nnct}(iv)}\\
&=\langle e\rangle^\#\tec\inl&
\text{\ref{nnct}(vi)}
\end{align*}
This completes the proof.
\end{proof}

\begin{thm}
\label{nncjn}
Let $H$~be a finitary, point-strong endofunctor.
Then the $H$-monoid~$RI$ above
is the initial iterative $H$-monoid.
\end{thm}

That is, for every $H$-monoid~$A$ there exists precisely one morphism
$h\colon RI\to A$ which is both a monoid homomorphism and a
homomorphism of $H$-algebras.

\proof
(1)\ 
Given an iterative $H$-monoid
\[\bar{a}\colon H\bar{A}\to\bar{A},\qquad \bar{i}\colon I\to\bar{A}
\qquad\text{and}\qquad \bar{m}\colon\bar{A}\otimes\bar{A}\to \bar{A}\]
we know that there is a unique $H$-algebra homomorphism
\[h\colon RI\to\bar{A}\qquad\text{with}\qquad h\tec\eta_I=\bar{i}.\]
It is our task to prove that $h$~preserves multiplication:\smallskip

\begin{equation}
\label{nnnrcjj}
\myvc{
\begin{psmatrix}[colsep=15mm,rowsep=10mm]
RI\otimes RI&RRI&RI\\
\bar{A}\otimes\bar{A}&&\bar{A}
\psset{arrows=->,nodesep=3pt}
\everypsbox{\scriptstyle}
\ncline{1,1}{1,2}^{\rule{3mm}{0mm}s^R}
\ncline{1,2}{1,3}^{\mu_I}
\ncline{1,3}{2,3}>{h}
\ncline{1,1}{2,1}<{h\otimes h}
\ncline{2,1}{2,3}_{\bar{m}}
\end{psmatrix}
}
\end{equation}\smallskip

\noindent(2)\
Recall $RI=\colim \mathsf{Eq}_I$ from Remark~\ref{nnct}.  We can
substitute~$\mathsf{EQ}_I$ by the category of all $e\colon W\to HW+I$
with
$W$~pointed. The argument is as in~(a2) of Theorem~\ref{nncs}. We
indicate pointing
by writing~$W_\bullet$ instead
of~$W$ (this stands for the notation~$(W,i^W)$). We know that
$\otimes$~is finitary, thus
\[RI\otimes RI=\colim \mathsf{Eq}_I\otimes \mathsf{Eq}_I\]
with the colimit cocone
\[e^\#\otimes e^\#\colon W_\bullet\otimes W_\bullet\to RI\otimes RI.\]
It is thus sufficient to prove that for every~$e$ the square\smallskip

\begin{equation}
\label{nnrcjd}
\myvc{
\begin{psmatrix}[colsep=15mm,rowsep=10mm]
W_\bullet\otimes W_\bullet&RI\otimes RI&RRI&RI\\
\bar{A}\otimes\bar{A}&&&\bar{A}
\psset{arrows=->,nodesep=3pt}
\everypsbox{\scriptstyle}
\ncline{1,1}{1,2}^{\rule{3mm}{0mm}e^\#\otimes e^\#}
\ncline{1,2}{1,3}^{s^R}
\ncline{1,3}{1,4}^{\mu_I}
\ncline{1,4}{2,4}>{h}
\ncline{1,1}{2,1}<{(he^\#)\otimes(he^\#)}
\ncline{2,1}{2,4}_{\bar{m}}
\end{psmatrix}
}
\end{equation}\smallskip

\noindent
commutes.

For every flat equation morphism $e\colon W\to HW+I$ recall $\langle
e\rangle$ from Notation~\ref{nnco}(iii) and put
\begin{equation}
\label{nnrcjdjp}
\bar{e}\equiv W\nsi{e} HW+I\nsi{HW+\bar{i}} HW+\bar{A}.
\end{equation}
We prove that~\eqref{nnrcjd} commutes by verifying that the two sides
of the square
are both the left-hand part of the solution $\bar{f}^\dag\colon
W\otimes W+W\to\bar{A}$
of the equation
morphism~$\bar{f}$ for
\[f=\langle e\rangle\colon W\otimes W+W\to H(W\otimes W+W)+I.\]\smallskip

\noindent(3)\
Proof of the upper passage of~\eqref{nnrcjd}:
\begin{equation}
\label{nnrcjt}
h\tec\mu_I\tec s^R\tec(e^\#\otimes e^\#)=\bar{f}^\dag\tec\inl\colon
W\otimes
W\to\bar{A}.
\end{equation}
Since $h\colon RI\to\bar{A}$
preserves solutions by Remark~\ref{nnct}(vii),
and since $f^\#$~is a solution
of $\eta_I\bullet f$ as mentioned in Remark~\ref{nnct}(iv), the
composite~$h\tec f^\#$ is a solution of the equation morphism
$\bar{f}$: indeed, $h$~takes $\eta_I\bullet f$ to~$\bar{f}$ due to
the diagram\smallskip

\[
\begin{psmatrix}[colsep=15mm,rowsep=10mm]
W\otimes W+W&&H(W\otimes W+W)+RI\\
&H(W\otimes W+W)+I\\
&&H(W\otimes W+W)+\bar{A}
\psset{arrows=->,nodesep=3pt}
\everypsbox{\scriptstyle}
\ncline{1,1}{1,3}\naput{\eta_I\bullet f}
\ncline{1,1}{2,2}\naput{f}
\ncline{2,2}{1,3}\nbput{\id+\eta_I}
\ncline{1,3}{3,3}>{\id+h}
\ncline{2,2}{3,3}\naput{\id+\bar{i}}
\ncarc[arcangle=-24]{1,1}{3,3}\nbput{\bar{f}}
\end{psmatrix}
\]\smallskip
Shortly, the triangle\smallskip

\begin{equation}
\label{nnrcjc}
\myvc{
\begin{psmatrix}[colsep=15mm,rowsep=10mm]
W\otimes W+W&RI\\
&\bar{A}
\psset{arrows=->,nodesep=3pt}
\everypsbox{\scriptstyle}
\ncline{1,1}{1,2}\naput{\quad f^\#}
\ncline{1,2}{2,2}>{h}
\ncline{1,1}{2,2}\nbput{\bar{f}^\dag}
\end{psmatrix}
}
\end{equation}\smallskip

\noindent commutes.

Observe that also the triangle\smallskip

\begin{equation}
\label{nnrcjp}
\myvc{
\begin{psmatrix}[colsep=15mm,rowsep=10mm]
RW\\
RRI&RI
\psset{arrows=->,nodesep=3pt}
\everypsbox{\scriptstyle}
\ncline{1,1}{2,2}\naput{\widehat{e^\#}}
\ncline{1,1}{2,1}<{Re^\#}
\ncline{2,1}{2,2}_{\mu}
\end{psmatrix}
}
\end{equation}\smallskip

\noindent
commutes since $RW$~is the free iterative algebra on $\eta_W\colon
W\to RW$ and the three morphisms above are homomorphisms of
$H$-algebras which are merged by~$\eta_W$---indeed:
\[\widehat{e^\#}\tec\eta_W=e^\#\]
as well as
\[\mu\tec Re^\#\tec\eta_W=\mu\tec\eta_{RI}\tec e^\#=e^\#.\]
Let us verify that
\begin{equation}
\label{nnrcjs}
s_{I,W}\tec (e^\#\otimes W)=(e*W)^\#.
\end{equation}
Indeed, for $e\colon W\to HW+I$ and $f=\eta_W$ form~$\hat{e}$
as in~\eqref{rcxxx} and observe
that $\hat{e}=\eta_W\bullet(e*W)$ holds. Now recall
from Equation~\eqref{nnrcojp}
that $s^R_{I,W}=\eta^\flat_W$. Thus, we have
\begin{align*}
s^R_{I,W}\tec(e^\#\otimes W)&=
\eta^\flat_W\tec(e^\#\otimes W)\\
&=\hat{e}^\dag&&\text{by \eqref{nnrco}}\\
&=\bigl(\eta_W\bullet(e*W)\bigr)^\dag\\
&=(e*W)^\#&&\text{by Remark \ref{nnct}(iv)}.
\end{align*}

We are now in the position to demonstrate~\eqref{nnrcjt}:\smallskip

\[
\begin{psmatrix}
W_\bullet\otimes W_\bullet&&(W_\bullet\otimes W_\bullet)+W_\bullet\\
RI\otimes W_\bullet&RW\\
RI\otimes RI&RRI&RI&\bar{A}
\psset{arrows=->,nodesep=3pt}
\everypsbox{\scriptstyle}
\ncline{1,1}{1,3}^{\inl}
\ncline{1,1}{2,2}>{(e* W_\bullet)^\#}
\ncline{1,1}{2,1}<{e^\#\otimes W}
\ncline{1,1}{2,1}>{\rule{1mm}{0mm}\eqref{nnrcjs}}
\ncline{2,1}{2,2}_{s^R}
\ncline{1,3}{3,3}>{f^\#}
\ncline{1,3}{3,3}<{%
\raisebox{5mm}{$\scriptstyle\eqref{nnrcjn}$\rule{10mm}{0mm}}}
\ncline{2,1}{3,1}<{RI\otimes e^\#}
\ncline{2,2}{3,2}<{Re^\#}
\ncline{2,2}{3,3}>{\widehat{e^\#}}
\ncline{1,3}{3,4}>{\bar{f}^\dag}
\ncline{3,1}{3,2}_{s^R}
\ncline{3,2}{3,3}_{\mu^I}
\ncline{3,3}{3,4}_{h}
\ncline{2,2}{3,3}<{\eqref{nnrcjp}\rule{2mm}{0mm}}
\ncline{3,3}{3,4}^{%
\raisebox{9mm}{$\scriptstyle\eqref{nnrcjc}$\rule{10mm}{0mm}}}
\end{psmatrix}
\]\smallskip

\comment{\item
Before continuing with our proof we observe two commutative squares.\smallskip

\begin{equation}
\label{nnnrcjs}
\myvc{
\begin{psmatrix}[colsep=15mm, rowsep=10mm]
HW\otimes W&H\bar{A}\otimes\bar{A}&H(\bar{A}\otimes\bar{A})\\
H(W\otimes W)&H(\bar{A}\otimes\bar{A})&H\bar{A}
\psset{arrows=->,nodesep=3pt}
\everypsbox{\scriptstyle}
\ncline{1,1}{1,2}^{H\bar{e}^\dag\otimes\bar{e}^\dag}
\ncline{1,2}{1,3}^{s^H}
\ncline{1,1}{2,1}<{s^H}
\ncline{1,3}{2,3}>{H\bar{m}}
\ncline{2,1}{2,2}_{H(\bar{e}^\dag\otimes\bar{e}^\dag)}
\ncline{2,2}{2,3}_{H\bar{m}}
\end{psmatrix}
}
\end{equation}\smallskip

\noindent
which is just the naturality square for~$s^H$ followed by~$H\bar{m}$,
and the other one is\smallskip

\begin{equation}
\label{nnrcjss}
\myvc{
\begin{psmatrix}[colsep=15mm, rowsep=10mm]
W=I\otimes W&\bar{A}\otimes W&\bar{A}\otimes\bar{A}\\
&I\otimes\bar{A}\\
W&&\bar{A}
\psset{arrows=->,nodesep=3pt}
\everypsbox{\scriptstyle}
\ncline{1,1}{1,2}^{\bar{i}\otimes W}
\ncline{1,2}{1,3}^{\bar{A}\otimes\bar{e}}
\ncline{1,1}{2,2}<{I\otimes\bar{e}}
\ncline{1,1}{3,1}<{W}
\ncline{2,2}{1,3}<{\bar{i}\otimes\bar{A}}
\ncline[arrows=-,doubleline=true]{2,2}{3,3}
\ncline{1,3}{3,3}>{\bar{m}}
\ncline{3,1}{3,3}_{e^\dag}
\end{psmatrix}
}
\end{equation}\smallskip

\noindent(4)\
We now prove the lower passage of~\eqref{nnrcjd}:
\[\bar{m}\bigl((h\tec\bar{e}^\#)\otimes(h\tec\bar{e}^\#)\bigr)=
\bar{f}^\dag\tec\inl.\]
Since the homomorphism $h\colon RI\to A$ preserves solutions and,
by~\ref{nnct}(iv), $e^\#$~is the solution of $\eta_I\bullet e$
in~$RI$, it follows that $\bar{e}^\dag=h\tec e^\#$. Thus. our task is
to prove
\begin{equation}
\label{nnrcjo}
\bar{m}\tec(\bar{e}^\dag\otimes\bar{e}^\dag)=\bar{f}^\dag\tec\inl.
\end{equation}
We first observe that we have a commutative diagram\smallskip

\begin{equation}
\label{nnnrcjo}
\myvc{
\begin{psmatrix}[colsep=22mm]
&\bar{A}\otimes{A}\\
W\otimes W&\bar{A}\otimes W&\bar{A}\\
&(H\bar{A}+\bar{A})\otimes W&\bar{A}\otimes\bar{A}+\bar{A}\\
(HW\otimes W)+W&(HW+\bar{A})\otimes W&
(H\bar{A}\otimes\bar{A})+(\bar{A}\otimes\bar{A})\\
H(W\otimes W)+W&H(\bar{A}\otimes\bar{A})+\bar{A}&
H\bar{A}+\bar{A}
\psset{arrows=->,nodesep=3pt}
\everypsbox{\scriptstyle}
\ncline{2,2}{1,2}<{\bar{A}\otimes\bar{e}^\dag}
\ncline{1,2}{2,3}>{\bar{m}}
\ncline{2,1}{2,2}^{\bar{e}^\dag W}
\ncline{2,1}{4,1}<{e\otimes W}
\ncline{2,1}{4,2}<{\bar{e}\otimes W}
\ncline{3,2}{2,2}>{[\bar{a},\bar{A}]\otimes W}
\ncline{5,1}{5,2}_{H(\bar{e}^\dag\otimes\bar{e}^\dag)+\bar{e}^\dag}
\ncline{3,3}{2,3}>{[\bar{m},\bar{A}]}
\ncline{4,2}{3,2}>{(H\bar{e}^\dag+\bar{A}\otimes W}
\ncline{3,2}{4,3}>{\id\otimes\bar{e}}
\ncline{4,3}{3,3}>{[\bar{a},\bar{A}]\otimes\bar{A}}
\ncline{4,1}{5,1}<{s^{H+W}}
\ncline{4,3}{5,3}>{H\bar{m}\tec s^H+\bar{m}}
\ncline{5,2}{5,3}_{H\bar{m}+\bar{A}}
\ncline{4,1}{4,2}_{(HW+\bar{i})\otimes W}
\ncline{4,1}{5,1}>{\rule{35mm}{0mm}\mbox{\small
\eqref{nnnrcjs} \&\ \eqref{nnrcjss}}}
\end{psmatrix}
}
\end{equation}\smallskip

\noindent
We now conclude~\eqref{nnrcjss} by using~\ref{nnco}(iii) (recall
$f=\langle e\rangle$):
}

\noindent(4)\
We shall prove for the lower passage of~\eqref{nnrcjd} that
\[\bar{m}\tec\bigl((h\tec\bar{e}^\#)\otimes (h\tec\bar{e}^\#)\bigr)=
\bar{f}^\dag\tec\inl\colon W\otimes W\to\bar{A}.\]
Since the homomorphism $h\colon RI\to A$ preserves solutions and,
by Remark~\ref{nnct}(iv),
$e^\#$~is the solution of $\eta_I\bullet e$
in~$RI$, it follows that $\bar{e}^\dag=h\tec e^\#$
(cf.~\eqref{nnrcjc}).
So we will prove
that
\[\bar{f}^\dag=\bigl[\bar{m}\tec(\bar{e}^\dag\otimes\bar{e}^\dag),
\bar{e}^\dag\bigr]\colon W\otimes W+W\to\bar{A}.\]
To see this it suffices to verify that the
following diagram\smallskip

\begin{equation}
\label{rcsss}
\myvc{
\begin{psmatrix}[colsep=35mm,rowsep=10mm]
W\otimes W+W&\bar{A}\\
H(W\otimes W+W)+\bar{A}&H\bar{A}+\bar{A}
\psset{arrows=->,nodesep=3pt}
\everypsbox{\scriptstyle}
\ncline{1,1}{1,2}^{[\bar{m}\tec(\bar{e}^\dag\otimes
\bar{e}^\dag),\bar{e}^\dag]}
\ncline{1,1}{2,1}<{\bar{f}}
\ncline{2,2}{1,2}>{[\bar{a},\bar{A}]}
\ncline{2,1}{2,2}_{\rule{8mm}{0mm}H[\bar{m}\tec(\bar{e}^\dag\otimes
\bar{e}^\dag),\bar{e}^\dag]+\bar{A}}
\end{psmatrix}
}
\end{equation}\smallskip

\noindent
commutes. In order to do so we consider the components of the upper
left-hand coproduct separately. For the right-hand component with
domain~$W$ we obtain\smallskip

\[
\begin{psmatrix}[colsep=25mm, rowsep=10mm]
W&&\bar{A}\\
HW+I&HW+\bar{A}\\
H(W\otimes W+W)+\bar{A}&&H\bar{A}+\bar{A}
\psset{arrows=->,nodesep=3pt}
\everypsbox{\scriptstyle}
\ncline{1,1}{1,3}^{\bar{e}^\dag}
\ncline{1,1}{2,1}<{e}
\ncline{1,1}{2,2}>{\bar{e}}
\ncline{2,1}{2,2}^{HW+\bar{i}}
\ncline{2,1}{3,1}<{H\inr+\bar{i}}
\ncline{2,2}{3,1}>{H\inr+\bar{A}}
\ncline{2,2}{3,3}>{H\bar{e}^\dag+\bar{A}}
\ncline{3,3}{1,3}>{[\bar{a},\bar{A}]}
\ncline{3,1}{3,3}_{\rule{8mm}{0mm}H[\bar{m}\tec(\bar{e}^\dag\otimes
\bar{e}^\dag),\bar{e}^\dag]+\bar{A}}
\end{psmatrix}
\]\smallskip

\noindent
This diagram commutes: the upper right-hand part commutes since
$\bar{e}^\dag$~is a solution of~$\bar{e}$, and all other inner parts
are obvious.

For the left-hand component of~\eqref{rcsss} we prove that the
following diagram\smallskip

\[
\begin{psmatrix}[colsep=21mm]
W\otimes W&\bar{A}\otimes\bar{A}&\bar{A}\\
(HW\otimes W)+W&(H\bar{A}+\bar{A})\otimes W&
(H\bar{A}\otimes\bar{A})+(\bar{A}\otimes\bar{A})\\
H(W\otimes W)+W&H(\bar{A}\otimes\bar{A})+\bar{A}&
H\bar{A}+\bar{A}\\
H(W\otimes W)+HW+\bar{A}&
H(\bar{A}\otimes\bar{A})+H\bar{A}+\bar{A}\\
H(W\otimes W+W)+\bar{A}&
H(\bar{A}\otimes\bar{A}+\bar{A})+\bar{A}&
H\bar{A}+\bar{A}
\psset{arrows=->,nodesep=3pt}
\everypsbox{\scriptstyle}
\ncline{1,1}{1,2}^{\bar{e}^\dag\otimes\bar{e}^\dag}
\ncline{1,2}{1,3}^{\bar{m}}
\ncline[arrows=-,doubleline=true]{3,3}{5,3}<{\mbox{(vii)}
\rule{17mm}{0mm}}
\ncline[arrows=-,doubleline=true]{3,3}{5,3}>{\rule{7mm}{0mm}
\mbox{(ii)}}
\ncline{1,1}{2,1}>{e\otimes W}
\ncline{2,3}{1,3}<{\mbox{(iv)}\rule{10mm}{0mm}
\bar{m}\tec([\bar{a},\bar{A}]\otimes
\bar{A})}
\ncline{2,1}{2,2}_{(H\bar{e}^\dag+\bar{i})\otimes W}
\ncline{2,2}{2,3}_{\id\otimes\bar{e}^\dag}
\ncline{2,1}{3,1}>{(s^H+W)\rule{43mm}{0mm}\mbox{(v)}}
\ncline{3,1}{4,1}>{\id+\bar{e}\rule{17mm}{0mm}%
\mbox{(vi)}}
\ncline{2,3}{3,3}<{H\bar{m}\tec s^H+\bar{m}}
\ncline{3,1}{3,2}_{H(\bar{e}^\dag\otimes\bar{e}^\dag)+\bar{e}^\dag}
\ncline{3,2}{3,3}_{H\bar{m}+\bar{A}}
\ncline{4,2}{3,2}>{\id+[\bar{a},\bar{A}]}
\ncline{4,1}{5,1}>{\can+\bar{A}\rule{13mm}{0mm}\mbox{(viii)}}
\ncline{4,2}{5,3}<{\mbox{(ix)}\rule{2mm}{0mm}}
\ncline{4,2}{5,2}>{\can+\bar{A}}
\ncline{4,1}{4,2}_{H(\bar{e}^\dag\otimes\bar{e}^\dag)+H\bar{e}^\dag+
\bar{A}}
\ncline{4,2}{5,3}>{H\bar{m}+[\bar{a},\bar{A}]}
\ncline{5,1}{5,2}^{H(\bar{e}^\dag\otimes
\bar{e}^\dag+\bar{e}^\dag)+\bar{A}}
\ncline{5,2}{5,3}^{H[\bar{m},\bar{A}]+\bar{A}}
\ncarc[arcangle=-10]{5,1}{5,3}_{H[\bar{m}\tec(\bar{e}^\dag\otimes
\bar{e}^\dag),\bar{e}^\dag]+\bar{A}}
\ncarc[arcangle=-70]{5,3}{1,3}<{[\bar{a},\bar{A}]}
\ncarc[arcangle=-70]{1,1}{5,1}>{\bar{f}\tec\inl}
\ncline{2,1}{3,1}<{\mbox{(i)}\rule{10mm}{0mm}}
\ncline{2,2}{1,2}<{\mbox{(iii)}\rule{12mm}{0mm}
[\bar{a},A]\otimes\bar{e}^\dag}
\end{psmatrix}
\]\bigskip\bigskip

\noindent
commutes: the left-hand part~(i) commutes by the definition of
$\bar{f}=\overline{\langle e\rangle}$ and
right-hand part~(ii) is obvious since
$\bar{A}$~is an $H$-monoid, i.e., $\bar{a}\tec H\bar{m}\tec
s^H=\bar{m}\tec (\bar{a}\otimes\bar{A})$.
For
part~(iii) observe that $(HW\otimes W)+W=(HW+I)\otimes W$ and use
Diagram~\eqref{rdd} and Equation~\eqref{nnrcjdjp}.
In part~(iv) we use the distributivity for the object in the
lower right-hand corner:
$(H\bar{A}+\bar{A})\otimes\bar{A}=H\bar{A}\otimes\bar{A}+
\bar{A}\otimes\bar{A}$---the commutativity is then obvious. We
postpone part~(v) to the end. Part~(vi) commutes by using~\eqref{rdd}.
Parts (vii) and~(viii) are trivial. We do not claim that
part~(ix) commutes, but it clearly does when post-composed
with~$[\bar{a},\bar{A}]$,
which suffices for the commutativity of the
outside of the diagram. Finally, it remains to
prove that part~(v) commutes: we consider the components of the
coproduct $(HW\otimes W)+W$ separately. For the right-hand component
we obtain the commutative diagram\smallskip

\begin{equation}
\label{nnrcjss}
\myvc{
\begin{psmatrix}[colsep=15mm,rowsep=10mm]
W=I\otimes W&\bar{A}\otimes W&\bar{A}\otimes\bar{A}\\
&I\otimes\bar{A}\\
W&&\bar{A}
\psset{arrows=->,nodesep=3pt}
\everypsbox{\scriptstyle}
\ncline{1,1}{1,2}^{\bar{i}\otimes W}
\ncline{1,2}{1,3}^{\bar{A}\otimes\bar{e}^\dag}
\ncline{1,1}{2,2}<{I\otimes\bar{e}^\dag}
\ncline[arrows=-,doubleline=true]{1,1}{3,1}
\ncline{2,2}{1,3}<{\bar{i}\otimes\bar{A}}
\ncline[arrows=-,doubleline=true]{2,2}{3,3}
\ncline{1,3}{3,3}>{\bar{m}}
\ncline{3,1}{3,3}_{\bar{e}^\dag}
\end{psmatrix}
}
\end{equation}\smallskip

\noindent
and for the left-hand one consider the diagram below\smallskip

\begin{equation}
\label{nnnrcjs}
\myvc{
\begin{psmatrix}[colsep=15mm,rowsep=10mm]
HW\otimes W&H\bar{A}\otimes\bar{A}&H(\bar{A}\otimes\bar{A})\\
H(W\otimes W)&H(\bar{A}\otimes\bar{A})&H\bar{A}
\psset{arrows=->,nodesep=3pt}
\everypsbox{\scriptstyle}
\ncline{1,1}{1,2}^{H\bar{e}^\dag\otimes\bar{e}^\dag}
\ncline{1,2}{1,3}^{s^H}
\ncline{1,1}{2,1}<{s^H}
\ncline{1,3}{2,3}>{H\bar{m}}
\ncline{2,1}{2,2}_{H(\bar{e}^\dag\otimes\bar{e}^\dag)}
\ncline{2,2}{2,3}_{H\bar{m}}
\ncline[arrows=-,doubleline=true]{2,2}{1,3}
\end{psmatrix}
}
\end{equation}\smallskip

\noindent
This completes the proof.\qed

\begin{thm}
\label{nncjjjp}
Let $H$~be a finitary, point-strong endofunctor. Then the
monoid~$TI$ from Corollary~\textnormal{\ref{cojp}}
is the initial completely iterative $H$-monoid.
\end{thm}

Notice that our proof below uses just the existence of the completely
iterative algebras for~$H$ (cf. Example~\ref{E:monad-T}) and not
finitariness of~$H$ directly.

\begin{proof}
Analogously to the preceding proof, for a completely iterative
$H$-monoid $(\bar{A},\bar{i},\bar{m},\bar{a})$ we know that there
exists a unique $H$-algebra homomorphism
\[h\colon TI\to\bar{A}\qquad\text{with}\qquad h\tec\eta^T_I=\bar{i}\]
and it is our task to prove that it preserves multiplication:\smallskip

\begin{equation}
\label{mon}
\myvc{
\begin{psmatrix}[colsep=15mm,rowsep=10mm]
TI\otimes TI&TTI&TI\\
\bar{A}\otimes\bar{A}&&\bar{A}
\psset{arrows=->,nodesep=3pt}
\everypsbox{\scriptstyle}
\ncline{1,1}{1,2}^{s^T}
\ncline{1,2}{1,3}^{\mu^T}
\ncline{1,1}{2,1}<{h\otimes h}
\ncline{1,3}{2,3}>{h}
\ncline{2,1}{2,3}_{\bar{m}}
\end{psmatrix}
}
\end{equation}\smallskip

\noindent
To prove this, we define an equation morphism in~$\bar{A}$:
\[e\colon TI\otimes TI=HTI\otimes TI+TI\to H(TI\otimes TI)+\bar{A}\]
such that both passages of~\eqref{mon} are solutions of~$e$
in~$\bar{A}$: put
\[e=s^H_{TI,TI}+h.\]
For the upper passage of~\eqref{mon} we need to verify that the square\smallskip

\begin{equation}
\label{mul}
\myvc{
\begin{psmatrix}[colsep=25mm,rowsep=10mm]
TI\otimes TI&\bar{A}\\
H(TI\otimes TI)+\bar{A}&H\bar{A}+\bar{A}
\psset{arrows=->,nodesep=3pt}
\everypsbox{\scriptstyle}
\ncline{1,1}{1,2}^{h\tec\mu^T\tec s^T}
\ncline{1,1}{2,1}<{e}
\ncline{1,2}{2,2}>{[\bar{a},\bar{A}]}
\ncline{2,1}{2,2}_{\rule{5mm}{0mm}H(h\tec\mu^T\tec s^T)+\bar{A}}
\end{psmatrix}
}
\end{equation}\smallskip

\noindent
commutes. Consider the components of $HTI\otimes TI+I\otimes TI$
separately. The left-hand component yields\smallskip

\[\begin{psmatrix}[colsep=20mm,rowsep=10mm]
H(TI\otimes TI)&HTTI&HTI&H\bar{A}\\
HTI\otimes TI\\
TI\otimes TI&TTI&TI&\bar{A}\\
H(TI\otimes TI)\\
H(TI\otimes TI)+\bar{A}&HTTI+\bar{A}&HTI+\bar{A}&H\bar{A}
\psset{arrows=->,nodesep=3pt}
\everypsbox{\scriptstyle}
\ncline{2,1}{1,1}<{s^H}
\ncline{1,1}{1,2}^{Hs^T}
\ncline{1,2}{1,3}^{H\mu^T}
\ncline{1,3}{1,4}^{Hh}
\ncarc[arcangle=-55]{2,1}{4,1}<{s^H}
\ncline{2,1}{3,1}>{\tau\otimes TI}
\ncarc[arcangle=70]{3,1}{5,1}>{e}
\ncline{3,1}{3,2}^{s^T}
\ncline{3,2}{3,3}^{\mu^T}
\ncline{3,3}{3,4}^{h}
\ncline{1,2}{3,2}>{\tau}
\ncline{1,3}{3,3}>{\tau}
\ncline{1,4}{3,4}>{\bar{a}}
\ncline{4,1}{5,1}<{\inl}
\ncline{5,1}{5,2}_{Hs^T+\bar{A}}
\ncline{5,2}{5,3}_{H\mu^T+\bar{A}}
\ncline{5,3}{5,4}_{Hh+\bar{A}}
\ncline{5,4}{3,4}>{[\bar{a},\bar{A}]}
\end{psmatrix}\]\smallskip

\noindent
The two upper right-hand squares commute since $\mu^T$ and~$h$ are
homomorphisms of algebras, the upper left-hand square
is Diagram~\eqref{ffflat} and the lower left-hand part
is obvious.
Thus, since the outside of the diagram clearly commutes we see that
its lower part~\eqref{mul} commutes when extended by the left-hand
injection $\tau_I\otimes TI$, as desired.

The right-hand component of~\eqref{mul} yields\smallskip

\[\begin{psmatrix}[colsep=20mm,rowsep=15mm]
I\otimes TI=TI\\
TI\otimes TI&TTI&TI&\bar{A}\\
\bar{A}&H(TI\otimes TI)+\bar{A}&&H\bar{A}+\bar{A}
\psset{arrows=->,nodesep=3pt}
\everypsbox{\scriptstyle}
\ncarc[arcangle=-55]{1,1}{3,1}<{h}
\ncline{1,1}{2,1}>{\eta^T\otimes TI}
\ncarc[arcangle=-10]{1,1}{2,2}>{\eta^TT}
\ncline[doubleline=true]{-}{1,1}{2,3}
\ncline{2,1}{2,2}_{s^T}
\ncline{2,1}{3,2}>{e}
\ncline{2,2}{2,3}_{\mu^T}
\ncline{2,3}{2,4}^{h}
\ncline{3,1}{3,2}_{\inr\rule{8mm}{0mm}}
\ncline{3,2}{3,4}_{H(h\tec\mu^T\tec s^T)+\bar{A}}
\ncline{3,4}{2,4}>{[\bar{a},\bar{A}]}
\end{psmatrix}\]\smallskip

\noindent
Since all other parts clearly commute we see that the right-hand
lower part~\eqref{mul} commutes when extended by the coproduct
injection $\eta^T_I\otimes TI$ as desired.
For the lower passage of~\eqref{mon} we verify that the square\smallskip

\begin{equation}
\label{low}
\myvc{
\begin{psmatrix}[colsep=20mm, rowsep=10mm]
TI\otimes TI&\bar{A}\otimes\bar{A}&\bar{A}\\
H(TI\otimes TI)+\bar{A}&H(\bar{A}\otimes\bar{A})+\bar{A}&
H\bar{A}+\bar{A}
\psset{arrows=->,nodesep=3pt}
\everypsbox{\scriptstyle}
\ncline{1,1}{1,2}^{h\otimes h}
\ncline{1,2}{1,3}^{\bar{a}}
\ncline{1,1}{2,1}<{e}
\ncline{2,3}{1,3}>{[\bar{a},\bar{A}]}
\ncline{2,1}{2,2}_{H(h\otimes h)+\bar{A}}
\ncline{2,2}{2,3}_{H\bar{m}+\bar{A}}
\end{psmatrix}
}
\end{equation}\smallskip

\noindent
commutes. The left-hand component of $TI\otimes TI=(HTI+I)\otimes TI=
(HTI\otimes TI)+TI$ (with coproduct injections $\tau\otimes TI$ and
$\eta^T\otimes TI)$
yields\smallskip

\[\begin{psmatrix}[colsep=21mm,rowsep=15mm]
HTI\otimes TI&H\bar{A}\otimes\bar{A}&H(\bar{A}\otimes\bar{A})&
H\bar{A}\\
H(TI\otimes TI)&TI\times TI&\bar{A}\otimes\bar{A}&\bar{A}\\
&H(TI\otimes TI)+\bar{A}&H(\bar{A}\otimes\bar{A})+\bar{A}&
H\bar{A}+\bar{A}
\psset{arrows=->,nodesep=3pt}
\everypsbox{\scriptstyle}
\ncline{1,1}{1,2}^{Hh\otimes h}
\ncline{1,2}{1,3}^{s^H}
\ncline{1,3}{1,4}^{H\bar{m}}
\ncline{1,1}{2,1}<{s^H}
\ncline{1,1}{2,2}<{\tau\otimes TI}
\ncline{1,2}{2,3}>{\bar{a}\otimes\bar{A}}
\ncline{1,4}{2,4}>{\bar{a}}
\ncline{2,1}{3,2}<{\inl}
\ncline{2,2}{3,2}>{e}
\ncline{3,4}{2,4}>{[\bar{a},\bar{A}]}
\ncline{3,2}{3,3}_{H(h\otimes h)+\bar{A}}
\ncline{3,3}{3,4}_{H\bar{m}+\bar{A}}
\ncline{2,2}{2,3}^{h\otimes h}
\ncline{2,3}{2,4}^{\bar{m}}
\end{psmatrix}\]\smallskip

\noindent
The upper right-hand square is Diagram~\eqref{abb}
and the middle one
commutes since $h$~is an $H$-algebra homomorphism. The left-hand part
is clear since $e=s^H+h$.
The lowest passage
in the last diagram is
\[\bar{a}\tec H\bar{m}\tec H(h\otimes h)\tec s^H= \bar{a}\tec
H\bar{m}\tec
s^H\tec(Hh\otimes h)\]
due to the naturality of~$s^H$. Thus, the outside of the diagram
commutes, showing that \eqref{low}~commutes when extended by
$\tau_I\otimes TI$.

Finally, the right-hand component
of~\eqref{low} is obvious:

\[\begin{psmatrix}[colsep=21mm,rowsep=15mm]
I\otimes TI=TI\\
\bar{A}&TI\otimes TI&\bar{A}\otimes\bar{A}&\bar{A}\\
&H(TI\otimes TI)+\bar{A}&H(\bar{A}\otimes\bar{A})+\bar{A}&
H\bar{A}+\bar{A}
\psset{arrows=->,nodesep=3pt}
\everypsbox{\scriptstyle}
\ncline{1,1}{2,1}<{h}
\ncline{1,1}{2,2}<{\eta^T\otimes TI}
\ncline{1,1}{2,3}<{\bar{i}\otimes h}
\ncline{1,1}{2,4}^{h}
\ncline{2,2}{2,3}_{h\otimes h}
\ncline{2,3}{2,4}_{\bar{a}}
\ncline{2,1}{3,2}<{\inl}
\ncline{2,2}{3,2}<{e}
\ncline{3,2}{3,3}_{H(h\otimes h)+\bar{A}}
\ncline{3,3}{3,4}_{H\bar{m}+\bar{A}}
\ncline{3,4}{2,4}>{[\bar{a},\bar{A}]}
\end{psmatrix}\]\smallskip

\noindent
The lower passage is~$h$ and so is the upper one:
\[\bar{a}\tec(\bar{i}\otimes h)=a\tec(i\otimes\bar{A})\tec (I\otimes
h)=I\otimes h.\]
This shows that \eqref{low} commutes when extended by
$\eta^T_I\otimes TI$, which completes the proof.
\end{proof}

\begin{cor}
\label{nncjjtc}
For every signature~\Sig\ the presheaf~$T_{\lambda,\Sig}$ is the
initial completely iterative $H_{\lambda,\Sig}$-monoid, and the
presheaf~$R_{\lambda,\Sig}$ is the initial iterative
$H_{\lambda,\Sig}$-monoid.\qed
\end{cor}

\comment{
This follows from Theorem~\ref{nncjjjp} since $T_{\lambda,\Sig}$~is
the free completely iterative $H_{\lambda,\Sig}$-algebra on~$V$ (cf.
Theorem~\ref{ddv}).

\begin{cor}
\label{nncjj}
$R_{\lambda,\Sig}$~is the initial iterative
$H_{\lambda,\Sig}$-monoid.
\end{cor}
}

Indeed, \SetF~is monoidally locally finitely presentable category
and $H_{\lambda,\Sig}$~is finitary and
point-strong by Corollary~\ref{hntt}.
Moreover, $\rr_{\lambda,\Sig}(V)=R_{\lambda,\Sig}$ by
Theorem~\ref{djp}, and $\TT_{\lambda,\Sig}(V)=T_{\lambda,\Sig}$ by
Theorem~\ref{nncjjjp}.

\section{Higher-Order Recursion Schemes}
\setcounter{equation}{0}
\label{kt}

\noindent We can reformulate (and slightly extend) higher-order recursion
schemes~\eqref{hrjj} categorically. Throughout this section we put
$\WW=\SetF$.
\begin{defi}
\label{tj}
A \textit{higher-order recursion scheme}
on a signature~\Sig\ (of ``terminals'') is a
presheaf morphism
\begin{equation}
\label{nott}
e\colon X\to F_{\lambda,\Sig}\otimes(X+V)
\end{equation}
where $X$~is a finitely presentable presheaf.
\end{defi}

\begin{rem}
\label{td}\hfill
\begin{enumerate}[(i)]
\item
The presheaf $F_{\lambda,\Sigma}\otimes (X+V)$
assigns to a context $\Gamma$ the set
$F_{\lambda,\Sigma}(X(\Gamma)+\Gamma)$
of finite $\lambda$-terms in contexts
$\overline{\Gamma}\subseteq X(\Gamma)+\Gamma$.

\item
  Although $F_{\lambda,\Sig}$ is the free
  $H_{\lambda,\Sig}$-algebra on $V$,
see
  Theorem~\ref{dss}, it is not in general true that
  $F_{\lambda,\Sig} \otimes Z$
  is the free $H_{\lambda,\Sig}$-algebra
  on a presheaf $Z$. For example, if
  $Z = V \times V$, then terms in
  $(F_{\lambda,\Sig} \otimes Z)(\Gamma)$ are
  precisely the finite $\lambda$-\Sig-terms whose free variables are
  substituted by pairs in $\Gamma \times \Gamma$. In contrast, the
  free $H_{\lambda,\Sig}$-algebra on $V \times V$ contains in context
  $\Gamma$ also terms such as $\lambda x.(x,y)$ for $y \in\Gamma$,
  that is, in variable pairs one member can be bound and one free.

\item
In the introduction we considered, for a given
context
\[\Gamma_{nt}=\{p_1,\dots,p_n\}\]
of ``nonterminals'', a system
of equations $p_i=f_i$, where $f_i$~is a $\lambda$-\Sig-term in some
context $\Gamma_0=\{x_1,\dots,x_k\}$. Let $X$~be the free presheaf
in $n$~generators $p_1,\dots,p_n$ of context~$\Gamma_0$
(a coproduct of $n$~copies of~$\FF(\Gamma_0, {-})$,
see Example~\ref{dd}(ii)).
Then the system of equations defines the unique
morphism
\[e\colon X\to F_{\lambda,\Sig}\otimes(X+V)\]
assigning to every~$p_i$ the right-hand side $f_i$ lying in
\[F_{\lambda,\Sig}(\Gamma_{nt}+\Gamma_0)\subseteq
F_{\lambda,\Sig}(X(\Gamma_0)+\Gamma_0).\]
Here we consider $F_{\lambda,\Sig}$ as an object
of~$\Fin(\Set,\Set)$.

\item
Conversely, every morphism~\eqref{nott} yields a system of
equations
$p_i=f_i$ as follows: let $\Gamma_0$~fulfill~\eqref{rdt} in
Definition~\ref{ndfdjj},
and define
$\Gamma_{nt}=X(\Gamma_0)$. The element $f_p=e_{\Gamma_0}(p)$ lies, for
every
nonterminal $p\in\Gamma_{nt}$, in
$F_{\lambda,\Sig}(\Gamma_{nt}+\Gamma_0)$.
We obtain a system of equations $p=f_p$
describing the
given morphism~$e$.

\item
We will use the presheaf~$R_{\lambda,\Sig}$
for our uninterpreted solutions
of recursion schemes:

A solution of the system of (formal) equations $p_i=f_i$ are rational
$\lambda$-\Sig-terms
$p_1^\dag,\dots,p_n^\dag$ making those equations identities
in~$R_{\lambda,\Sig}(\Gamma_0)$ when we substitute in~$f_i$ the
$\lambda$-\Sig-terms~$p^\dag_j$
for the nonterminals~$p_j$ ($j=1,\dots,n$). This is
expressed by the Definition~\ref{ntd} below.

\item The general case of ``equation morphisms'' as considered
in~\cite{AMV1} is (for the endofunctor $H_{\lambda,\Sigma}$) a
morphism of type $e:X\to {\mathbb{R}}_{\lambda,\Sigma}(X+V)$.  We
see that every higher-order recursion scheme gives an equation
morphism via the inclusion $F_{\lambda,\Sig} \subto R_{\lambda,\Sig}$
and the strength of the monad $\mathbb{R}_{\lambda,\Sig}$ (but not
necessarily conversely). Our solution theorem below is an application
of the general result of~\cite{AMV1}.
\end{enumerate}
\end{rem}

\begin{defi}
\label{ntd}
A \textit{solution} of a higher-order
recursion scheme $e\colon X\to
F_{\lambda,\Sig}\otimes(X+V)$ is a morphism $e^\dag\colon X\to
R_{\lambda,\Sig}$ such that the square below, where $j\colon
F_{\lambda,\Sig}\to R_{\lambda,\Sig}$ denotes the embedding, commutes:
$$
\xymatrix@C+2pc@R-.5pc{
  X
  \ar[r]^{e^\dagger}
  \ar[d]_e
  &
  R_{\lambda,\Sig}
  \\
  F_{\lambda,\Sig} \otimes (X + V)
  \ar[d]_{j \otimes (X + V)}
  \\
  R_{\lambda,\Sig} \otimes (X + V)
  \ar[r]_{\rule{2mm}{0mm}R_{\lambda,\Sig} \otimes [e^\dagger, i]}
  &
  R_{\lambda,\Sig} \otimes R_{\lambda,\Sig}
  \ar[uu]_{m}
}
$$
\comment{
\vspace{1pt}
\[\begin{psmatrix}[colsep=25mm]
X&R_{\lambda,\Sig}\\
F_{\lambda,\Sig}\otimes(X+V)&\\
R_{\lambda,\Sig}\otimes(X+V)&R_{\lambda,\Sig}\otimes R_{\lambda,\Sig}
\psset{arrows=->,nodesep=3pt}
\everypsbox{\scriptstyle}
\ncline{1,1}{1,2}^{e^\dag}
\ncline{1,1}{2,1}<{e}
\ncline[arrows=H->,bracketlength=-.6,tbarsize=-3mm]{2,1}{3,1}<{j
  \otimes (X+V)}
\ncline{3,2}{1,2}>{m^R}
\ncline{3,1}{3,2}_{\rule{4mm}{0mm}%
R_{\lambda,\Sig}\otimes [e^\dag,i^R]\rule{3mm}{0mm}}
\end{psmatrix}\]}
\end{defi}

\smallskip
\begin{exa}
\label{tt}
The equation for the
fixed-point combinator (see Example~\ref{jj}) with
$\Sig=\emptyset$
defines~$e$ whose domain is
the terminal presheaf~1, that is,
$e\colon1\to F_\lambda\otimes(1+V)$.
The solution
$e^\dag\colon1\to R_\lambda$
assigns to the unique element of~1 the tree~\eqref{rnjj}.
\end{exa}

\begin{rem}
\label{tp}
Recursion schemes such as $p_1=p_1$ make no sense---and they certainly
fail to have a unique solution. In general, we want to avoid
right-hand sides of the form~$p_i$. A recursion scheme is called
\textit{guarded} if no right-hand side lies in~$\Gamma_{nt}$.
(Theorem~\ref{hncss} below shows that no other restrictions are
needed.)
Guardedness can be formalized as follows: since
\[R_{\lambda,\Sig}=H_{\lambda,\Sig}(R_{\lambda,\Sig})+V\qquad
\text{with injections $\varrho_V$ and $i$}\]
by Remark~\ref{copr}, we have (see Remark~\ref{hndj}(i))
\[R_{\lambda,\Sig}\otimes(X+V)\cong H_{\lambda,\Sig}(R_{\lambda,\Sig})
\otimes(X+V)+X+V\]
with coproduct injections $\varrho_V\otimes\id_{X+V}$ and
$i\otimes\id_{X+V}$.
Then $e$~is guarded if its extension $(j \otimes (X +V))\tec e\colon X\to
R_{\lambda,\Sig}\otimes(X+V)$
factorizes through the embedding of
the first and third summand of this coproduct:
\end{rem}

\begin{defi}
\label{hncs}
A higher-order recursion scheme $e\colon
X\to F_{\lambda,\Sig}\otimes(X+V)$
is called \textit{guarded} if
$(j\otimes (X+V))\tec e$~factorizes through
\[\bigl[\varrho\otimes\id,(i\otimes\id)\tec\inr\bigr]\colon
H_\lambda(R_{\lambda,\Sig})\otimes(X+V)+V\to R_{\lambda,\Sig}\otimes
(X+V).\]
\end{defi}

\begin{thm}
\label{hncss}
Every guarded higher-order recursion scheme has a unique solution.\qed
\end{thm}

\begin{trivlist}
\item[\hspace{\labelsep}\textbf{Remark.}]
In Definition~\ref{tj} we restricted higher-order recursion schemes to
have $F_{\lambda,\Sig}$ in their codomain. This corresponds well to
the classical notion of recursion schemes as explained in
Remark~\ref{td}. Moreover, this leads to a simple presentation of the
interpreted semantics in Section~\ref{hp} below. However,
Theorem~\ref{hncss} remains valid if we
replace~$F_{\lambda,\Sig}$ by~$R_{\lambda,\Sig}$ in
Definition~\ref{tj} and define solution by $e^\dagger = m \cdot
R_{\lambda,\Sig} \otimes [e^\dagger, i] \cdot e$. This extends the
notion of a
higher-order recursion scheme~\eqref{hrjj}
to allow the right-hand sides~$f_i$ to be
rational $\lambda$-\Sig-terms. We
shall prove Theorem~\ref{hncss} working with higher-order
schemes of the form
$e\colon X\to R_{\lambda,\Sig}\otimes(X+V)$, $X$ finitely
presentable. We call $e$ guarded if it factorizes through
$[\varrho\otimes\id,(i\otimes\id)\tec\inr]$.
\end{trivlist}

\begin{proof}
Let us apply
the monad~$\rr_{\lambda,\Sig}$ and its point-strength~$s^R$ from
Theorem~\ref{nncs}. We construct for
every higher-order recursion scheme $e\colon
X\to\rr_{\lambda,\Sig}(V)\otimes(X+V)$ a rational equation
morphism $\bar{e}\colon X\to\rr_{\lambda,\Sig}(X+V)$ in the sense
of~\cite{AMV1} as follows:
\[\bar{e}\equiv X\nsi{e}\rr_{\lambda,\Sig}(V)\otimes(X+V)
\nsi{s^R_{(V,\id)(X+V,\inr)}}\rr_{\lambda,\Sig}(X+V).\]
From the guardedness of~$e$ we conclude that $\bar{e}$~is guarded in
the sense of~\cite{AMV1}, that is, $\bar{e}$~factorizes through the
summand
$H_{\lambda,\Sig}\rr_{\lambda,\Sig}(X+V)+V$ of
$\rr_{\lambda,\Sig}(X+V)$,
see Remark~\ref{copr}.
Indeed, this follows from the
following diagram\smallskip

\[\begin{psmatrix}[colsep=12mm,rowsep=10mm]
&H_{\lambda,\Sig}(R_{\lambda,\Sig})\otimes(X+V)+V&&
H_{\lambda,\Sig}(\rr_{\lambda,\Sig}(X+V))+V\\
X&R_{\lambda,\Sig}\otimes(X+V)&&\rr_{\lambda,\Sig}(X+V)
\psset{arrows=->,nodesep=3pt}
\everypsbox{\scriptstyle}
\ncline{1,2}{1,4}^{H_{\lambda,\Sig}s^R+\id}
\ncline[linestyle=dashed]{2,1}{1,2}
\ncline{2,1}{2,2}_{e}
\ncline{2,2}{2,4}_{s^R}
\ncline[arrows=H->,bracketlength=-.6,tbarsize=-3mm]{1,2}{2,2}
\ncline[arrows=H->,bracketlength=-.6,tbarsize=-3mm]{1,4}{2,4}
\end{psmatrix}\]\smallskip

\noindent
Consequently, by Theorem~4.5 in~\cite{AMV1} there exists a unique
solution~$\bar{e}^\dag$ of~$\bar{e}$ with respect to the
monad~$\rr_{\lambda,\Sig}$---this means that there exists
a unique morphism
$\bar{e}^\dag\colon X\to\rr_{\lambda,\Sig}(V)$
such that the outside of
the diagram below commutes:\smallskip

\[\begin{psmatrix}[colsep=5mm,rowsep=10mm]
X&&&&&&\rr_{\lambda,\Sig}(V)\\
\rr_{\lambda,\Sig}(V)\otimes(X+V)&&&&&
\rr_{\lambda,\Sig}(V)\otimes\rr_{\lambda,\Sig}(V)&\\
\rr_{\lambda,\Sig}(X+V)&&&&&&
\rr_{\lambda,\Sig}(V)(\rr_{\lambda,\Sig}(V))
\psset{arrows=->,nodesep=3pt}
\everypsbox{\scriptstyle}
\ncline{1,1}{1,7}^{\bar{e}^\dag}
\ncline{1,1}{2,1}<{e}
\ncline{2,1}{2,6}_{\rr_{\lambda,\Sig}(\id_V)\otimes[\bar{e}^\dag,i]}
\ncline{2,6}{1,7}<{m}
\ncline{2,1}{3,1}<{s_{(V,\id)(X+V,\inr)}^R}
\ncline{2,6}{3,7}<{s_{(V,\id)(R_{\lambda,\Sig},i)}^R}
\ncline{3,1}{3,7}_{\rr_{\lambda,\Sig}[\bar{e}^\dag,i]}
\ncline{3,7}{1,7}>{\mu_V}
\end{psmatrix}\]\smallskip

\noindent
Since the right-hand triangle is Equation~\eqref{multi}
and the lower
square commutes by the naturality of~$s^R$, we see that
$\bar{e}^\dag$ is a solution of~$e$ in the sense of
Definition~\ref{ntd} iff
$\bar{e}^\dag$~is a solution of~$\bar{e}$ in the sense of~\cite{AMV1}.
This proves that $e$~has a unique solution.
\end{proof}

\begin{rem}
\label{solution}
Analogously we could define a solution of a higher-order recursion
scheme of the form $e\colon X\to T_{\lambda,\Sig}\otimes(X+V)$ in the
initial completely iterative monoid~$T_{\lambda,\Sig}$, see
Theorem~\ref{hntss}. (Here we, moreover, do not need to assume that
$X$~is finitely presentable.)
And guardedness means here that $e$~factorizes through the summand
$H_{\lambda,\Sig}T_{\lambda,\Sig}\otimes(X+V)+V$ of
$T_{\lambda,\Sig}\otimes(X+V)$ (cf. Remarks \ref{R:monad-T}
and~\ref{hndj}(i)).
Every such guarded scheme has a unique solution
in~$T_{\lambda,\Sig}$. The proof is completely analogous to the
previous one for~$R_{\lambda,\Sig}$ but using Corollary~3.8
in~\cite{AAMV} in lieu of Theorem~4.5 from~\cite{AMV1}.
\end{rem}

\begin{rem}
\label{npdv}
Notice that the definitions and results in this section generalize to
the setting as considered in Section~\ref{snc} (see
Assumption~\ref{nncj}). Simply replace~$H_{\lambda,\Sig}$ by the
finitary functor~$H$, the monoid~$F_{\lambda,\Sig}$ by the initial
$H$-monoid~$F$; this exists and is given by the free $H$-algebra
on~$V$, see~\cite{FPT}. Further replace
the monoid~$R_{\lambda,\Sig}$ by the initial
iterative $H$-monoid~$RI$ (cf. Theorem~\ref{nncjn}), and
$T_{\lambda,\Sig}$~by the initial completely iterative
$H$-monoid (cf.
Theorem~\ref{nncjjjp}).
\end{rem}

\section{Interpreted Solutions}
\setcounter{equation}{0}
\label{hp}

\noindent In the present section we prove that every Scott model of
$\lambda$-calculus as a \CPO, $D$, with fold and unfold operations
can be used as a model of higher-order recursion. Following M.~Fiore
et al~\cite{FPT} we work with the presheaf~$\langle D,D\rangle$ which
to a context~$\Gamma$ assigns the set of all continuous functions
from~$D^\Gamma$ to~$D$. We prove that every higher-order recursion
scheme has a least solution in~$\langle D,D\rangle$.

We denote by $\cpob$ the cartesian closed
category of posets with directed joins
and continuous functions. Thus least elements are not assumed; if
they exist we use~$\perp$ for them.

\begin{asm}
\label{nnpj}
We assume that a Scott model $D$ of
$\lambda$-calculus is given, i.e., a CPO with $\bot$ and with an
embedding-projection pair
\begin{equation}
\label{eq:5.1}
\fold: \cpob(D,D)\triangleleft D : \unfold.
\end{equation}
Moreover, for the given signature~\Sig\ of terminals we also assume
that continuous operations
$$
\sigma^D:D^n\to D
\quad
\mbox{for every $n$-ary $\sigma$ in $\Sigma$}
$$
are given.
\end{asm}

\begin{notation}
\label{nnpd}
We define a presheaf~$\langle D,D\rangle$ by
$$
\spitze{D}\Gamma=\cpob(D^\Gamma,D).
$$
\end{notation}

\begin{rem}
\label{nnpt}\hfill
\begin{enumerate}[(a)]
\item
Observe that elements
of $\spitze{D}$ can always be interpreted in $D$: the
above function $\fold\colon\spitze{D}1\to D$ yields obvious
functions $\fold_\Gamma\colon\spitze{D}\Gamma\to D$ for all
contexts~$\Gamma$ via induction:
define~$\fold_{\Gamma+1}$ by\smallskip

\[
\begin{psmatrix}[colsep=20mm,rowsep=10mm]
\makebox[0pt][r]{$\langle D,D\rangle(\Gamma+1)\cong\mbox{}$}\cpob
(D^\Gamma\times D,D)&D\\
\cpob(D^\Gamma,D^D)&\cpob(D^\Gamma,D)
\psset{arrows=->,nodesep=3pt}
\everypsbox{\scriptstyle}
\ncline{1,1}{1,2}^{\fold_{\Gamma+1}}
\ncline{2,2}{1,2}>{\fold_\Gamma}
\ncline{1,1}{2,1}<{\curry}
\ncline{2,1}{2,2}_{\cpob(D^\Gamma,\fold)}
\end{psmatrix}
\]\smallskip

\noindent
where $\curry\colon\cpob(D^\Gamma \times D, D) \to
\cpob(D^\Gamma, D^D)$
is the currification.
\comment{
$$
\spitze{D}(\Gamma+1)\cong\cpob(D^\Gamma\times D,D)
\nsi{\curry}
\cpob(D^\Gamma,[D,D])
\nsi{\cpob(D^\Gamma,\fold)}
\cpob(D^\Gamma,D)
\nsi{\fold_\Gamma}
D
$$
}

\item
The projections $D^\Gamma\to D$ are continuous functions, that is,
elements of~$\langle D,D\rangle\Gamma$. This defines a natural
pointing
\[\iota\colon V\to\langle D,D\rangle\]
of the presheaf~$\langle D,D\rangle$.
\end{enumerate}
\end{rem}

\begin{rem}
\label{nnpc}
The presheaf $\spitze{D}$ is an $H_{\lambda,\Sigma}$-monoid.
Indeed, application and abstraction are naturally obtained
from~(\ref{eq:5.1}), see~\cite{FPT}.
The monoid structure
$$
m:\spitze{D}\otimes\spitze{D}\to\spitze{D}
$$
can be described directly by using the coend
formula~(\ref{eq:3.2}) or indirectly:

\begin{enumerate}[(a)]
\item
For the direct description consider the component
of~$m_\Gamma$ corresponding, for an element
$f\in\Set(\overline{\Gamma},\cpob(D^\Gamma,D))$, to the
injection
$$
\injection_f:
\cpob(D^{\overline{\Gamma}},D)
\to
\int^{\overline{\Gamma}}
\Set(\overline{\Gamma},\cpob(D^\Gamma,D))
\bullet
\cpob(D^{\overline{\Gamma}},D).
$$
Observe that $f$~yields a continuous function $\tilde{f}\colon
D^\Gamma\to D^{\bar{\Gamma}}$ defined by $\tilde{f}(x)=f({-})(x)$ for
all
$x\in\Gamma$. We define the component $m_\Gamma\tec\injection_f$
of~$m_\Gamma$ by\smallskip

\begin{equation}
\label{eq:5.2}
\rule{30mm}{0mm}
\myvc{
\begin{psmatrix}[colsep=10mm,rowsep=10mm]
\makebox[0pt][r]{$\langle D,D\rangle\otimes\langle
D,D\rangle(\Gamma)=\mbox{}$}\displaystyle\int^{\bar{\Gamma}}
\Set(\bar{\Gamma},\cpob(D^\Gamma,D))\bullet \cpob(D^{\bar{\Gamma}},D)&
\cpob(D^\Gamma,D)\makebox[0pt][l]{$\mbox{}=\langle
D,D\rangle\Gamma$}\\
\CPO(D^{\bar{\Gamma}},D)
\psset{arrows=->,nodesep=3pt}
\everypsbox{\scriptstyle}
\ncline{1,1}{1,2}^{\rule{22mm}{0mm}m_\Gamma}
\ncline{2,1}{1,1}<{\injection_f}
\ncline{2,1}{1,2}>{g\mapsto g\tec\tilde{f}}
\end{psmatrix}
}
\end{equation}\smallskip

\item
There is a much more elegant way of obtaining the
monoid structure of $\spitze{D}$. From results
of Steve Lack~\cite{lack} we see that the monoidal
category $(\Set^\FF,\otimes, V)$ has the following
monoidal action $*$ on $\cpob$: given $X$ in $\Set^\FF$
and $C$ in $\cpob$, we put
$X*C=\int^\Gamma X(\Gamma)\bullet C^\Gamma$.
Moreover, extending the above notation to pairs
$C$, $C'$ of CPO's and defining
$\langle C,C'\rangle\Gamma=\cpob(C^\Gamma,C')$
we obtain a presheaf with a natural isomorphism
$$
\Set^\FF (X,\langle C,C'\rangle)
\cong
\cpob(X*C,C').
$$
As observed by George Janelidze and Max Kelly~\cite{janelidze+kelly}
this yields an enriched category whose hom-objects are
$\langle C,C'\rangle$. In particular, $\spitze{D}$ receives
a monoid structure. It is tedious but not difficult to prove
that this monoid structure is given by~(\ref{eq:5.2})
above and it forms an $H_{\lambda,\Sigma}$-monoid
(cf.~Definition~\ref{hntc}).
\end{enumerate}
\end{rem}

\begin{notation}
\label{nnpp}
We denote by
$$
\sem{{-}}:F_{\lambda,\Sigma}\to\spitze{D}
$$
the unique $H_{\lambda,\Sigma}$-monoid homomorphism (see Theorem~\ref{hnts}). For
every finite term $t$ in context $\Gamma$ we thus obtain its
interpretation as a continuous function
$\sem{t}_\Gamma:D^\Gamma\to D$
\end{notation}

\begin{rem}
\label{nnps}
What is our intuition of an interpreted solution of higher order
recursion scheme
$e:X\to F_{\lambda,\Sigma}\otimes (X+V)$ in the presheaf
$\spitze{D}$? This should be an interpretation of $X$-terms
in $\spitze{D}$ via a natural transformation
$$
e^\dag:X\to\spitze{D}
$$
with the following property: Given an $X$-term $x$ in context
$\Gamma$, then $e_\Gamma$ assigns to it an element
$e_\Gamma(x)$ of
$(F_{\lambda,\Sigma}\otimes (X+V))(\Gamma)$, that is, a finite
term $t\in F_{\lambda,\Sigma}(\overline{\Gamma})$ for some
$\overline{\Gamma}\subseteq X(\Gamma)+\Gamma$. We request that
the solution assigns to $x$ the same value
$e^\dag_\Gamma (x):D^\Gamma\to D$ that we obtain from the
interpretation $\sem{t}_1$ of the given term by substituting
the $\overline{\Gamma}$-variables using
$[e^\dag,\iota]:X+V\to\spitze{D}$. This substitution is given by
composing $\sem{{-}}\otimes [e^\dag,\iota]$ with the monoid structure
of $\spitze{D}$. This leads to the following
\end{rem}

\begin{defi}
\label{nnpss}
Given a higher-order recursion scheme
$e:X\to F_{\lambda,\Sigma}\otimes (X+V)$
by an \textit{interpreted solution} is meant a presheaf
morphism $e^\dag:X\to\spitze{D}$ such that the square below commutes:
\begin{equation}
\label{eq:int_solution}
\vcenter{
  \xymatrix@C+2pc@R-.5pc{
    X \ar[r]^{e^\dag}
    \ar[d]_e
    &
    \spitze{D}
    \\
    F_{\lambda,\Sig}\otimes (X + V)
    \ar[r]_{\sem{{-}}\otimes[e^\dag,\iota]}
    &
    \spitze{D}\otimes\spitze{D}
    \ar[u]_m
    }
  }
\comment{
\myvc{
\begin{psmatrix}
X&\spitze{D}\\
F_{\lambda,\Sigma}\otimes (X+V)&\spitze{D}\otimes\spitze{D}
\psset{arrows=->,nodesep=3pt}
\everypsbox{\scriptstyle}
\ncline{1,1}{1,2}^{e^\dag}
\ncline{2,2}{1,2}>{m}
\ncline{1,1}{2,1}<{e}
\ncline{2,1}{2,2}_{\sem{{-}}\otimes [e^\dag,\iota]}
\end{psmatrix}
}}
\end{equation}
\end{defi}

\begin{thm}
\label{nnpo}
Every higher-order recursion scheme has a least interpreted solution in
$\spitze{D}$ in the pointwise ordering of $\SetF(X,\langle
D,D\rangle)$.
\end{thm}

\begin{proof}
Observe that $\Set^\FF(X,\spitze{D})$ is a CPO with $\bot$.
Therefore it is sufficient to prove that the endomap of
$\Set^\FF(X,\spitze{D})$ given by
\begin{equation}
\label{nnrpc}
s\mapsto m\cdot (\sem{{-}}\otimes [s,\iota])\cdot e
\end{equation}
is continuous, then we can use the
Kleene Fixed-Point Theorem. Observe that the function~\eqref{nnrpc} is
a composite\smallskip

\begin{equation}
\label{nnrpp}
\myvc{
\begin{psmatrix}[colsep=30mm,rowsep=10mm]
\SetF(X,\langle D,D\rangle)&\SetF(X,\langle D,D\rangle)\\
\SetF(X+V,\langle D,D\rangle)&
\SetF(F_{\lambda,\Sig}\otimes(X+V),\langle D,D\rangle)
\psset{arrows=->,nodesep=3pt}
\everypsbox{\scriptstyle}
\ncline{1,1}{1,2}^{s\mapsto m\tec(\sem{{-}}\otimes[s,\iota])\tec e}
\ncline{1,1}{2,1}<{[{-},\iota]}
\ncline{2,2}{1,2}>{{-}\tec e}
\ncline{2,1}{2,2}_{z\mapsto m\tec(\id_{F_{\lambda,\Sig}}\otimes
z)\rule{10mm}{0mm}}
\end{psmatrix}
}
\end{equation}\smallskip

\noindent
where the vertical arrows are obviously continuous. It remains to
prove the continuity of
\[z\mapsto m\tec(\id\otimes z)\qquad\text{for $z\colon Z\to\langle
D,D\rangle$}\]
where $Z=X+V$ (but this structure of~$Z$ plays no role).

We use the coend formula~\eqref{eq:3.2}. It is obvious what the
components of
\[(\id\otimes z)_\Gamma\colon \int^{\bar{\Gamma}}
\Set(\bigl(\bar{\Gamma},Z(\Gamma)\bigr)\bullet\langle D,D\rangle
\bar{\Gamma}\to \int^{\bar{\Gamma}}
\Set\bigl(\bar{\Gamma},\langle D,D\rangle\Gamma\bigr)\bullet \langle
D,D\rangle\bar{\Gamma}\]
are: the $\bar{\Gamma}$-component composed with the coproduct
injection of $u\in\Set(\bar{\Gamma},Z(\Gamma))$ yields the coproduct
injection of
\[f=z_\Gamma\tec u\in\Set\bigl(\bar{\Gamma},\langle
D,D\rangle\Gamma\bigr).\]
Combined with the component of~$m_\Gamma$, see~\eqref{eq:5.2}, this
yields the components of~$(m\tec(\id\otimes z))_\Gamma$ as follows:
for every context~$\bar{\Gamma}$ the component composed with the
coproduct injection~$\injection_u$ of~$u$ is the map
\[g\mapsto g\tec\widetilde{z_\Gamma\tec u}\qquad\text{for
$g\in\langle D,D\rangle\bar{\Gamma}$.}\]\smallskip

\noindent We are ready to prove that the lower horizontal arrow
of~\eqref{nnrpp} is continuous. Suppose
\[z=\bigsqcup_{k\in K} z^k\]
is a directed join in the pointwise ordering of $\SetF(X+V,\langle
D,D\rangle)$. In order to prove the equality
\[m\tec(\id\otimes z)=\bigsqcup_{k\in K} m\tec(\id\otimes z^k)\]
in~\SetF, we choose an arbitrary context~$\Gamma$ and prove that
\[m_\Gamma\tec(\id\otimes z_\Gamma)=\bigsqcup_{k\in K} m_\Gamma\tec
(\id\otimes z^k_\Gamma)\]
holds in~\Set. For that use the fact that the injection maps, for all
contexts~$\bar{\Gamma}$ and all $u\in\Set(\bar{\Gamma},Z(\Gamma))$,
\[\langle D,D\rangle\bar{\Gamma}\nsi{\injection_u}
\Set\bigl(\bar{\Gamma},Z(\Gamma)\bigr)\bullet\langle D,D\rangle
\bar{\Gamma}\nsi{\injection_{\bar{\Gamma}}} \int^{\bar{\Gamma}}
\Set\bigl(\bar{\Gamma},Z(\Gamma)\bigr)\bullet\langle
D,D\rangle\bar{\Gamma}\]
form a collectively epimorphic cocone. Thus, it is sufficient to
prove that for every context~$\bar{\Gamma}$ and every $g\in\langle
D,D\rangle\bar{\Gamma}$ we have
\[g\tec\widetilde{z_\Gamma\tec u}=\bigsqcup_{k\in K}
g\tec\widetilde{z^k_\Gamma\tec u}.\]
Since $g$~is continuous from~$D^{\bar{\Gamma}}$ to~$D$, we just need
to verify
\[\widetilde{z_\Gamma\tec u}=\bigsqcup_{k\in K}
\widetilde{z^k_\Gamma\tec u}.\]
From the pointwise ordering we clearly get $z_\Gamma\tec
u=\bigsqcup_{k\in K} z^k_\Gamma\tec u$, thus, we only need to observe
the continuity of the map $f\mapsto\tilde{f}$, and this follows from
the
coordinate-wise ordering of~$D^\Gamma$.
\end{proof}

\section{Conclusions}
\setcounter{equation}{0}
\label{p}

\noindent We proved that guarded higher-order
recursion schemes have a unique
uninterpreted solution, i.e., a solution as a rational
$\lambda$-\Sig-term. And they also have the least
interpreted solution for interpretations based on
Scott's models of $\lambda$-calculus as CPO's with
continuous operations for all ``terminal'' symbols
of the recursion scheme.

Following M.~Fiore \textit{et al}~\cite{FPT} we worked in the
category~\SetF\ of sets in context, that is, covariant presheaves on
the category~\FF\ of finite sets and functions.
A presheaf is a
set dependent on a context (a finite set of variables). For
every signature~\Sig\ of ``terminal'' operation symbols it was proved
in~\cite{FPT} that the presheaf~$F_{\lambda,\Sig}$ of all finite
$\lambda$-\Sig-terms is the initial
$H_{\lambda,\Sig}$-monoid. This means that $F_{\lambda,\Sig}$~has
(i)~the $\lambda$-operations (of abstraction and application)
together with the operations given by~\Sig\ rendering an
$H_{\lambda,\Sig}$-algebra, (ii)~the operation expressing
simultaneous substitution rendering a monoid in the category of
presheaves, and (iii)~these two structures are canonically related.
And $F_{\lambda,\Sig}$~is the initial presheaf with such
structure. R.~Matthes and T.~Uutalu~\cite{MU} showed
that the
presheaf~$T_{\lambda,\Sig}$ of finite and infinite
$\lambda$-\Sig-terms is also an $H_{\lambda,\Sig}$-monoid. Here we
proved that this is
the initial completely iterative $H_{\lambda,\Sig}$-monoid.
And its subpresheaf~$R_{\lambda,\Sig}$
of all rational $\lambda$-\Sig-terms is the
initial iterative $H_{\lambda,\Sig}$-monoid.
We used that last
presheaf in our uninterpreted semantics of recursion schemes.

Our approach was based on
untyped $\lambda$-calculus. The ideas in the typed
version are quite analogous. If $S$~is the set of all types, then we
form the full subcategory~\FF\ of~$\Set^S$ of finite $S$-sorted sets
and consider presheaves in~$(\Set^S)^{\FF}$---the latter category is
equivalent to that of
finitary endofunctors of the category~$\Set^S$.
The
definition of~$H_{\lambda,\Sig}$ is then completely
analogous to the untyped case, and one can form the
presheaves~$F_{\lambda,\Sig}$ (free algebra on~$V$),
$T_{\lambda,\Sig}$ (free completely iterative algebra)
and~$R_{\lambda,\Sig}$ (free iterative algebra). Each of them is
a monoid, in fact, an $H_{\lambda,\Sig}$-monoid in the sense
of~\cite{FPT}. Moreover, every guarded higher-order recursion scheme
has a unique solution in~$R_{\lambda,\Sig}$.
The interpreted semantics can be built up on a
\CPO-enriched cartesian closed category (as our model of typed
$\lambda$-calculus) with additional continuous morphisms for all
terminals. The details of the typed version are more involved, and
we leave them for future work.

Related results on higher-order substitution can be found e.g. in
\cite{MU} and~\cite{P}.

In future work we will, analogously as in~\cite{MM}, investigate the
relation of uninterpreted and interpreted solutions.

\end{document}